%% file: NC-Wick.tex
\documentclass[a4paper,12pt]{article}
\input{layout}

\input{commands}

\newarrow{DoubleArrow}<--->

\title{Wick rotation for quantum field theories on degenerate Moyal space(-time)}
\author{
	Harald Grosse$^a$,\;\; Gandalf Lechner$^a$\footnote{Supported by FWF project P22929--N16 ``Deformations of quantum field theories''.},\;\;
	Thomas Ludwig$^{b,c}$\footnote{Supported by IMPRS, Max-Planck-Institute for Mathematics in the Sciences, Leipzig}
	,\;\; Rainer Verch$^c$\\
	\vspace*{-1mm}\\
	{\small $^a$ Department of Physics, University of Vienna, 1090 Vienna, Austria}\\
	{\small $^b$ Max-Planck-Institute for Mathematics in the Sciences, 04103 Leipzig, Germany}\\
	{\small $^c$ Institute for Theoretical Physics, University of Leipzig, 04009 Leipzig, Germany}
}
\date{November 30, 2011}

\begin{document}
\maketitle

\begin{abstract}
	In this paper the connection between quantum field theories on flat noncommutative space(-times) in Euclidean and Lorentzian signature is studied for the case that time is still commutative. By making use of the algebraic framework of quantum field theory and an analytic continuation of the symmetry groups which are compatible with the structure of Moyal space, a general correspondence between field theories on Euclidean space satisfying a time zero condition and quantum field theories on Moyal Minkowski space is presented (``Wick rotation''). It is then shown that field theories transferred to Moyal space(-time) by Rieffel deformation and warped convolution fit into this framework, and that the processes of Wick rotation and deformation commute.
\end{abstract}

\input{introduction}
\input{nets}

\input{symmetries}
\input{warping}

\input{conclusion}



\end{document}

%% file: layout.tex
\usepackage{amsmath}
\usepackage{amsfonts}
\usepackage{amssymb}
\usepackage{mathrsfs}
\usepackage[small,nohug]{diagrams}
\usepackage[pdfborder={0 0 0}]{hyperref}
\usepackage{color}
\usepackage[T1]{fontenc}

\newtheorem{theorem}{Theorem}[section]
\newtheorem{proposition}[theorem]{Proposition}
\newtheorem{lemma}[theorem]{Lemma}

\newtheorem{definition}[theorem]{Definition}

\newenvironment{proof}{\medskip \noindent {\em Proof:}}{\hfill $\square$ \\[2mm] \indent}

\numberwithin{equation}{section}

\newlength{\dinwidth}
\newlength{\dinmargin}

%% file: commands.tex
\usepackage{oldgerm}
\usepackage{color}
\usepackage{bbm}

\renewcommand{\mathbb}[1]{\mathbbm{#1}}

\newcommand{\Cl}{\mathbbm{C}}
\newcommand{\Rl}{\mathbb{R}}
\newcommand{\Nl}{\mathbb{N}}

\definecolor{lightgray}{rgb}{0.8,0.8,0.8}

\newcommand{\Om}{\Omega}
\newcommand{\om}{\omega}

\newcommand{\te}{\theta}

\newcommand{\La}{\Lambda}

\newcommand{\A}{\mathcal{A}}

\newcommand{\B}{\mathcal{B}}
\newcommand{\C}{\mathcal{C}}

\newcommand{\M}{\mathcal{M}}
\newcommand{\K}{\mathcal{K}}

\newcommand{\NN}{\mathcal{N}}

\newcommand{\Hil}{\mathcal{H}}

\newcommand{\E}{\mathcal{E}}

\newcommand{\DD}{\mathcal{D}}

\newcommand{\UU}{{\mathscr U}}

\newcommand{\Slic}{\mathscr{S}}
\newcommand{\Zyl}{\mathscr{Z}}

\newcommand{\Ss}{\mathscr{S}}   

\newcommand{\OOO}{\mathscr{O}}

\newcommand{\phiti}{\tilde{\phi}}
\newcommand{\phihat}{\hat{\phi}}

\newcommand{\LGpo}{\mathcal{L}_+^\uparrow}

\newcommand{\LG}{\mathcal{L}}

\newcommand{\PG}{\mathcal{P}}

\newcommand{\SO}{\text{SO}}

\newcommand{\OO}{O}

\newcommand{\GL}{\text{GL}}

\newcommand{\supp}{\text{supp}\,}

\newcommand{\dom}{\mathrm{dom}\,}



%% file: introduction.tex
\section{Introduction}

Non-commutative spacetime draws its main motivation from the idea that it might avoid spacetime singularities, which are one of the inevitable consequences of classical gravity in general relativity \cite{HawkingEllis:1973}. Seen in this light, one can view non-commutative quantum field theory \cite{DoplicherFredenhagenRoberts:1995, DouglasNekrasov:2001, Szabo:2003}, {\em i.e.}, the investigation of quantum fields on non-commutative spacetime, as an initial step towards a quantum theory of gravity. However, one may also take the point of view that quantum fields, in particular interacting quantum fields, ought to be investigated on non-commutative spacetime in their own right. Motivation for this point of view can be drawn from indications that spacetime non-commutativity has a ``smoothing'' effect on short-distance singularities of quantum fields. This facilitates their renormalization and, eventually (or so one hopes), might render the construction of interacting quantum fields in physical spacetime dimension possible. In the light of the notorious difficulties to establish existence of interacting quantum fields in four-dimensional ``commutative'' Minkowski spacetime, such a possibility is a viable incentive.

One of the lessons of constructive quantum field theory is that it is often more efficient to attempt the construction of interacting quantum field theories on Euclidean space rather than on Minkowski spacetime \cite{GlimmJaffe:1987,Rivasseau:1991}. Then, once a Euclidean quantum field theory has been constructed, usually in terms of Schwinger functions on Euclidean space fulfilling a certain set of assumptions, one can rely on the celebrated Osterwalder-Schrader theorem \cite{OsterwalderSchrader:1973,OsterwalderSchrader:1975} which ensures that the Schwinger functions can be analytically continued from Euclidean time to physical time, resulting in Wightman functions on Minkowski spacetime. Thus, under suitable conditions, the Osterwalder-Schrader theorem establishes a one-to-one correspondence between Euclidean quantum field theories formulated in terms of Schwinger functions, and quantum field theories on Minkowski spacetime in the framework of Wightman functions. The purpose of the present work is to establish a similar correspondence for quantum field theories on Moyal-deformed Euclidean space, and Moyal-deformed Minkowski spacetime, respectively, albeit under the fairly restrictive assumption of ``keeping time commutative'' ({\em i.e.,} a degenerate Moyal-deformation).
\\
\\
The defining feature of Moyal space are the commutation relations
\begin{align}\label{eq:MoyalRelations}
	[X_\mu,X_\nu]
	=
	i \theta_{\mu \nu} \cdot 1
\end{align}
between its coordinates $X_0,...,X_{d-1}$, where $\te$ is a real, antisymmetric $(d\times d)$-matrix. As is well known, these commutation relations can be reformulated in terms of a noncommutative star product $f,g\mapsto f\times_\te g$ between suitable functions $f,g$ on $\Rl^d$, or, somewhat more generally, as a deformation of the product of an algebra carrying an automorphic action $\tau$ of the translation group $\Rl^d$, formally defined as\footnote{See \cite{Rieffel:1992} for a rigorous treatment of this product in the context of $C^*$-algebras, and \cite{LechnerWaldmann:2011} for a recent generalization to locally convex algebras.}
\begin{align}\label{eq:DeformedProduct}
	A\times_\te B
	&=
	\frac{1}{(2\pi)^d}\int dp\,dx\,e^{i(p,x)}\tau_{\te p}(A)\tau_x(B)
	\,.
\end{align}
Here $(\,\cdot\,,\,\cdot\,)$ is a non-degenerate properly normalized bilinear form on $\Rl^d$, such as the Euclidean  or Lorentzian inner product $(p,x)^\E = \sum_{\mu=0}^{d-1} p_\mu x_\mu$ respectively $(p,x)^\M = -p_0x_0+\sum_{\mu=1}^{d-1} p_\mu x_\mu$.

Deformed products of this type can be used to define field theories on Moyal space in several ways. For example, a popular approach consists in defining a field theory via a classical action functional, and then replacing all pointwise products in this action by star products, {\em i.e.}, use \eqref{eq:DeformedProduct} for an algebra of functions on $\Rl^d$ with pointwise product and $\tau$ the natural action of translations. Via perturbative quantization one then arrives at a corresponding field theory on noncommutative $\Rl^d$, either in a Euclidean or Minkowskian setting\footnote{Our account of that on the following lines is necessarily incomplete and should be seen as a sample as concerns literature cited.}.

In the Minkowskian situation, quite a number of works have investigated how Feynman rules and operator ordering prescriptions are to be modified in order to render a perturbatively constructed interacting theory unitary, starting out from a Moyal-deformed version of a free quantum field on Minkowski spacetime \cite{LiaoSibold:2002-1,LiaoSibold:2002-2,BahnsDoplicherFredenhagenPiacitelli:2002,BahnsDoplicherFredenhagenPiacitelli:2003,DenkSchweda:2003}. This question is highly non-trivial since the deformed product $\times_\theta$ destroys commutativity of the fields at spacelike separated spacetime points, resulting in a host of technical and conceptual difficulties.

Another approach to transferring field theories from $\Rl^d$ to Moyal spacetime consists in deforming field operators $\phi\mapsto\phi_\te$ directly \cite{GrosseLechner:2007}. With this method, known as ``warped convolution'' \cite{BuchholzSummers:2008}, any quantum field theory on Minkowski spacetime can be deformed to a Moyal-Minkowski spacetime, and the residual localization properties can be analyzed in a model-independent manner \cite{GrosseLechner:2008,BuchholzLechnerSummers:2010}.
At the mathematical level, this approach amounts to certain representations of algebras with deformed products of the form \eqref{eq:DeformedProduct}.

Other work is devoted to non-commutative external field scattering and the question if Moyal-deformed quantum field theories can be related to --- as yet, somewhat hypothetical --- quantum fields over certain Lorentzian spectral triples \cite{PaschkeVerch:2004,BorrisVerch:2008,Verch:2011}.
\\
\\
In the Euclidean setting, Feynman rules have first been worked out by Filk \cite{Filk:1996}. Initially progress was slow because of the phenomenon of UV/IR mixing in Moyal-deformed theories \cite{MinwallaVanRaamsdonkSeiberg:1999}. It was later found that this problem can be circumvented by modifying the Euclidean Lagrangean for a self-interacting theory through adding a term akin to an oscillator potential \cite{GrosseWulkenhaar:2005}. From this point on, Moyal-deformed Euclidean quantum field theory has made significant progress, particularly in constructing interacting models \cite{Rivasseau:2007, GrosseWulkenhaar:2009, Wang:2011}. Analogously as for usual, ``undeformed''  quantum field theory, it appears  that also for Moyal-deformed quantum field theory the construction of interacting models is often more efficient in the Euclidean setting than in the Minkowskian world.

On the other hand, for Moyal-deformed quantum field theory there are as yet no results like an Osterwalder-Schrader theorem which would allow one to pass, {\em e.g.}, from a model theory constructed in the Euclidean setting to a quantum field theory (Moyal-deformed, or other) on Minkowski spacetime. In view of the fact that the Wick rotation relies substantially on the covariance and locality properties of quantum field theories \cite{StreaterWightman:1964}, one might actually expect serious problems in generalizing it to a Moyal-deformed setting, where usually both these properties are significantly weakened. In particular in the case of an invertible deformation matrix $\theta$, the continuation of the $n$-point functions of a Moyal-Minkowski space field theory to imaginary times does not match correspondingly $\te$-deformed Euclidean Schwinger functions \cite{Bahns:2009}. The fact that the connection between the Euclidean and Minkowskian world is subtle indeed in the noncommutative setting is also witnessed by the ongoing discussion in models with oscillator terms in the action \cite{FischerSzabo:2008, Zahn:2010}.

Therefore, one may have doubts if a relation between Euclidean and Minkowskian field theories on Moyal space(-time) can be established at all by any sort of Wick-rotation. However, it turns out that a tight relation does exist between the algebras generated by Minkowskian and Euclidean versions of a Moyal-deformed quantum field, as long as ``time remains commutative''. This means that $\theta$, viewed as a linear map on $\mathbb{R}^d$, has a non-trivial kernel containing a (timelike) unit vector $e$. In order to establish that relation --- which is the purpose of this paper --- it has turned out instrumental to adopt the framework of the operator-algebraic approach to quantum field theory \cite{Haag:1996}.

Let us recall here some of the basic ingredients of the operator-algebraic approach; the full details required for this work will be described in the main body of the text. A quantum field theory on $d$-dimensional Minkowski spacetime is described by an operator algebra $\mathcal{M}$ with an action $\alpha^\M$ of the proper, orthochronous Poincar\'e group by automorphisms on $\mathcal{M}$. Quite importantly, $\mathcal{M}$ has a local sub-structure, {\em i.e.}, it is formed by local subalgebras subject to the conditions of isotony and locality, and the automorphisms $\alpha^\M_g$ are required to act covariantly with respect to that local sub-structure. We will describe the local sub-structure, whose precise form depends on the localization properties of the theory (wherein {\em e.g.} a theory on Moyal-Minkowski spacetime differs from a theory on usual Minkowski spacetime) in more detail in the next section. Furthermore, it is assumed that the theory possesses a vacuum state. Formally, given some localization region $O$ in $\mathbb{R}^d$, and denoting by $\mathcal{M}(O)$ the associated local subalgebra, one can think of $\mathcal{M}(O)$ as collecting all bounded functions of quantum fields $\phi^{\mathcal{M}}(x)$ with $x \in O$.\footnote{Strictly speaking, one has to pass to field operators smeared against test-functions to render this statement viable.} Hence, the $\mathcal{M}(O)$ are algebras formed by the quantum field observables which can be observed within the spacetime region $O$ (supposing the quantum fields under consideration are observable).

In an analogous manner, one can describe a Euclidean quantum field theory in the operator-algebraic setting. Here, a Euclidean quantum field theory is described by an operator algebra $\mathcal{E}$ together with an action $\alpha^\E$ of the Euclidean group by automorphisms of $\mathcal{E}$. Again, it is important that $\mathcal{E}$ has a local sub-structure, and that the $\alpha_g^{\mathcal{E}}$ act covariantly with respect to that local sub-structure. A further key ingredient is a reflection-positive functional $\sigma:\E\to\Cl$, which satisfies a specific positivity condition with respect to a chosen direction $e$ in Euclidean space. (Again, we will give a precise definition in the main body of this article.)

In this operator-algebraic setting, Schlingemann \cite{Schlingemann:1999-1}, drawing on results established by Fr\"ohlich, Osterwalder and Seiler \cite{FrohlichOsterwalderSeiler:1983}, and by Klein and Landau \cite{KleinLandau:1982}, has shown that from any Euclidean theory --- shortly denoted by $\mathcal{E}$ --- one can obtain a unique Minkowskian theory --- denoted by $\mathcal{M}$ in shorthand notation --- provided it is assumed that the Euclidean theory fulfills the {\em time-zero condition}. This condition demands essentially that there exists a sub-algebra $\E_0$ of operators in $\mathcal{E}$ which are ``localized on the hyperplane $e^\perp$'' and that $\E_0$ is large enough to generate the full observable algebra $\mathcal{E}$ via the Euclidean action $\alpha^\E$.

The just described setting of deriving a Minkowskian quantum field theory $\mathcal{M}$ from a Euclidean theory $\mathcal{E}$ by operator-algebraic methods is our starting point. Let $e$ be the unit vector in $\mathbb{R}^d$ entering in the definition of the reflection-positive functional $\sigma$. One can apply Schlingemann's procedure of ``algebraic Wick rotation'', obtaining a Minkowskian theory $\mathcal{M}$. Then let $\theta$ be a deformation matrix such that $\theta e = 0$, {\em i.e.}, $e$ lies in the kernel of $\theta$. By warped convolution, one can proceed to a corresponding deformed observable algebra $\widetilde{\mathcal{M}_\theta}$, describing the Minkowskian theory on Moyal space. This algebra still carries an automorphic action $\tilde{\alpha}^{\mathcal{M}}$ of the subgroup $\PG_\te(d)$ of the Poincar\'e group compatible with the deformation matrix $\te$, and the action is covariant with respect to the residual local sub-structure of $\widetilde{\mathcal{M}_\theta}$.

However, instead of first ``Wick-rotating'' a given Euclidean theory, and then ``Moyal-deforming'' the resulting Minkowskian theory, one can apply these procedures in the reverse order. Starting from a Euclidean theory $\mathcal{E}$, one can first obtain its deformed version $\mathcal{E}_\theta$ by deforming the operator product in $\mathcal{E}$ to the Rieffel-product \eqref{eq:DeformedProduct}. Since time is commutative, the same deformation matrix $\te$ can be used here for both signatures. As we will show, Schlingemann's ``Wick-rotation'' in the operator-algebraic setting can be extended to cover Moyal-deformed theories such as $\E_\te$, and yields a Minkowskian counterpart, denoted by ${\mathcal{M}_\theta}$, with an automorphic action $\alpha^\M$ of $\mathcal{P}_\theta(d)$, acting covariantly with respect to a residual local sub-structure of $\mathcal{M}_\theta$.

As $\widetilde{\M}_\te$, $\tilde{\alpha}^\M$, the data $\M_\te,\alpha^\M$ describe a Minkowskian Moyal-deformed version of the initial Euclidean theory $\E$, and we will show that $\mathcal{M}_\theta$ and $\widetilde{\mathcal{M}_\theta}$ are actually indistinguishable: There is an isomorphism $\varphi: \widetilde{\mathcal{M}_\theta} \to \mathcal{M}_\theta$ which intertwines the actions $\widetilde{\alpha}^{\mathcal{M}}$ and $\alpha^{\mathcal{M}}$. Moreover, $\varphi$ respects the residual local sub-structures of the theories $\mathcal{M}_\theta$ and $\widetilde{\mathcal{M}_\theta}$. Thus, the result of this work can be summarized by stating that the following diagram commutes:
\begin{diagram}
	&& \E && \\
	& \ldTo^{\text{\small deformation}} && \rdTo^{\text{\small Wick rotation}} & \\
	\E_\te &&&& \M \\
	\dTo^{\text{\small Wick rotation}} &&&& \dTo_{\text{\small deformation}} \\
	{\M_\te} &&\rDoubleArrow_\varphi && \widetilde{\M_\te}
\end{diagram}
Loosely speaking, this diagram says that, provided one assumes the time-zero condition and ``commutative time'', the Wick-rotation relation between a Euclidean and Minkowskian field theory carries over to their respective Moyal-deformed counterparts.

The content of this work is organized as follows. In Section 2, we give a precise formulation of the operator-algebraic setting of quantum field theories on Minkowski spacetime, and on Euclidean space, basically adopted from \cite{Schlingemann:1999-1, FrohlichOsterwalderSeiler:1983,KleinLandau:1982}. We also illustrate the connection between Euclidean quantum field theory in terms of Schwinger functions, or Euclidean path integrals, and the operator-algebraic setting. In Section~3, we show that Schlingemann's procedure of ``Wick-rotation'' in the algebraic setting can in fact be generalized to Euclidean theories having the structure a Moyal-deformed Euclidean theory would have; this refers in particular to weaker covariance properties. We draw largely on results of \cite{FrohlichOsterwalderSeiler:1983, KleinLandau:1982, JorgensenOlafsson:1999} in order to achieve that generalization. In Section 4 we apply the result of Section 3 and show that a Moyal-deformed Euclidean theory $\mathcal{E}_\theta$, with commutative time, can be algebraically Wick-rotated into a Minkowskian theory $\widetilde{\mathcal{M}_\theta}$, and we establish the commutative diagram indicated above. Conclusion and outlook complete this paper in Section 5.

%% file: nets.tex
\section{Euclidean and Minkowski nets of observables}\label{section:nets}

To prepare the ground for a model-independent analysis of the connection between deformed field theories on $\Rl^d$ with Euclidean and Lorentzian signature, we will in this section introduce a suitable operator-algebraic framework and discuss the relevant symmetry groups appearing in this context. For both signatures, the observables will be described by operator algebras ($C^*$-algebras and von Neumann algebras) $\E$ (Euclidean space) respectively $\M$ (Minkowski spacetime). Both $\E$ and $\M$ are required to have a covariant net structure in the sense of the following definition.

\begin{definition}\label{def:net}
	Given a manifold $M$, a family $\OOO$ of subsets of $M$, and a group $G$ of point transformations of $M$ which leaves $\OOO$ invariant, a $G$-covariant net $(\A,\OOO,\alpha)$ on $M$ is defined as the following structure. $\A$ is a map from $\OOO$ to $C^*$-algebras $\A(\OO)$ (respectively von Neumann algebras acting on a common Hilbert space $\Hil$), such that
	\begin{align}\label{isotony}
		\A(\OO_1)\subset\A(\OO_2)
		\;\,
		\text{for}
		\;\,
		\OO_1\subset \OO_2\,.
	\end{align}
	The smallest $C^*$-algebra (respectively von Neumann algebra) containing all $\A(\OO)$, $\OO\in\OOO$, is also denoted $\A$, and $\alpha$ is an automorphic action of $G$ on $\A$, such that
	\begin{align}\label{covariance}
		\alpha_g(\A(\OO)) = \A(g\OO)
		\,,\qquad
		g\in G\,,\;\OO\in\OOO\,.
	\end{align}
\end{definition}

In the standard algebraic formulation of quantum field theory \cite{Araki:1999, Haag:1996}, a quantum field theory on Minkowski space is viewed as a net $(\M,\OOO,\alpha^\M)$ on $\Rl^d$, with $\OOO$ the family of open bounded regions in $\Rl^d$, which is covariant under the Poincar\'e group $\PG(d)$, and furthermore satisfies locality (causality) in the usual form
\begin{align}\label{locality}
	[\M(\OO_1),\,\M(\OO_2)]=\{0\}
	\qquad\text{for}\qquad
	\OO_1\subset\OO_2'\,,
\end{align}
where $\OO_2'$ denotes the causal complement of $\OO_2$ with respect to the Minkowski metric $\eta=\rm{diag}(-1,+1,...,+1)$ on $\Rl^d$.

Euclidean field theories, on the other hand, are usually formulated in terms of their Schwinger functions or a measure generating these $n$-point functions (path integral) \cite{GlimmJaffe:1987}. Nonetheless, an operator-algebraic description is possible also in this case, viewing a Euclidean field theory as a net $(\E,\OOO,\alpha^\E)$ of $C^*$-algebras on $\Rl^d$, again with $\OOO$ as the open bounded subsets of $\Rl^d$, which is covariant under the Euclidean group $E(d)$ \cite{Schlingemann:1999-1}.

Concretely, the construction of a Euclidean net from a path integral can be sketched as follows. Adopting the framework of Glimm and Jaffe \cite{GlimmJaffe:1987}, we assume that there exists a probability measure $\mu$ on the space $\Ss':=\Ss'(\Rl^d,\Rl)$ of real distributions on Schwartz space, generating the Schwinger distributions $S_n\in\Ss'(\Rl^{nd})$, $n\in\Nl$, as its moments,
\begin{align}
	S_n(f_1\otimes ... \otimes f_n) = \int_{\Ss'} d\mu(\varphi)\,\varphi(f_1)\cdots\varphi(f_n)
	\,,\qquad
	f_1,...,f_n\in\Ss(\Rl^d)\,.
\end{align}
The measure $\mu$ is assumed to satisfy the standard properties \cite{GlimmJaffe:1987} of regularity, Euclidean invariance under the canonical action $\beta$ of the Euclidean group on $\Ss'$, and Osterwalder-Schrader reflection positivity \cite{OsterwalderSchrader:1973}.

To make contact with the algebraic formulation, we consider the ``Euclidean Hilbert space''
\begin{align}
	\Hil^\E := L^2(\Ss'\to\Cl,d\mu)\,,
\end{align}
which carries a representation $U^\E$ of $E(d)$,
\begin{align}
	(U^\E(g)F)(\varphi) := F(\beta_g\,\varphi)\,,\qquad F\in\Hil^\E,\;\varphi\in\Ss'\,.
\end{align}
The invariance and regularity of $\mu$ imply that $U^\E$ is unitary and weakly continuous. Defining the support $\supp F$ of a function $F\in\Hil^\E$ as the smallest closed subset $\Delta\subset\Rl^d$ such that $F(\varphi)=0$ for all $\varphi\in\Ss'(\Rl^d)$ with $\supp\varphi\cap\Delta=\emptyset$, it is also clear that $\supp U^\E(g)F=g\,\supp F$ for any $g\in E(d)$. Furthermore, the constant function $\Om^\E(\varphi):=1$, $\varphi\in\Ss'$, is a unit vector in $\Hil^\E$ because $\mu$ is a probability measure, and $\Om^\E$ is invariant under $U^\E$.

Concerning the Euclidean observables, the ``field'' acts as an unbounded multiplication operator $A(f)$, $f\in\Ss(\Rl^d)$, on the Euclidean Hilbert space via
\begin{align}
	(A(f) F)(\varphi) := \varphi(f)\cdot F(\varphi)\,,\qquad F\in\Hil^\E\,.
\end{align}
Typical elements of the Euclidean $C^*$-algebra are bounded multiplication operators such as $e^{i A(f)}$; and in fact, the Schwinger distributions can be recovered from the expectation values $\int d\mu(\varphi)\,\exp iA(f)$ under certain analyticity requirements on the measure \cite{GlimmJaffe:1987}. We take here the abelian $C^*$-algebra
\begin{align}
	\E := \bigvee\{\exp iA(f)\,:\,f\in \Ss_\Rl(\Rl^d)\}
\end{align}
as our definition of the Euclidean observable algebra. This algebra naturally has a net structure by taking $\E(\OO)\subset\E$ as the algebra of all operators multiplying with functions having support in $\OO\subset\Rl^d$. Furthermore, the Euclidean symmetry acts on this net by the automorphisms $\alpha^\E_g(A) := U^\E(g)A\,U^\E(g)^{-1}$, and this action is covariant in the sense that $\alpha^\E_g(\E(\OO))=\E(g\OO)$. Thus a measure generating a family of Schwinger distributions gives rise to an $E(d)$-covariant net $(\E,\OOO,\alpha^\E)$ of $C^*$-algebras on $\Rl^d$.
\\
\\
In the following, we will consider $G$-covariant nets $(\A,\OOO,\alpha)$ on $\Rl^d$ which have an additional feature. Fixing a unit vector $e\in\Rl^d$ (timelike for Minkowski signature), we write $e^\perp\subset\Rl^d$ for the hyperplane orthogonal to $e$, and introduce the {\em time zero algebras}
\begin{align}\label{time-zero-data}
	\A_0(K)
	:=
	\bigcap_{\OO\supset K} \A(\OO)
	\,,\qquad K\subset e^\perp
	\,.
\end{align}
If the time zero algebras and the action $\alpha$ generate the original net, {\em i.e.} if
\begin{align}
	\A(\OO)
	&=
	\bigvee_{K\subset e^\perp,\,g\in G\atop gK\subset \OO} \alpha_g(\A_0(K))
	\,,\qquad
	\OO\in\OOO\,,
	\label{time-zero}
\end{align}
we will speak of a net satisfying the {\em time zero condition}. Depending on the context, the symbol $\bigvee$ denotes either the $C^*$- or von Neumann algebra generated.

This assumption seems to be a strong condition and restricts the class of models we can analyze in our current setting. It is known that many free field models and models with polynomial self-interaction in dimension $d\leq3$ \cite{GlimmJaffe:1987} satisfy the time zero condition. However, it is unclear if there exist interacting quantum field theories in four space-time dimensions which comply with this condition. Since our construction method relies on the time-zero condition as in \cite{Schlingemann:1999-1}, we will need to assume it here.
\\
\\
As is well known, the road from Euclidean to Minkowski space field theories passes through a vacuum representation. Such a representation can be obtained from a $\PG(d)$-covariant net $(\M,\OOO,\alpha^\M)$ with the help of a vacuum state $\om$, and from an $E(d)$-covariant net $(\E,\OOO,\alpha^\E)$ with a reflection positive functional $\sigma$. This concept involves the reflection $r_e$ which inverts $e$, that is $r_e:x\mapsto x-2(e,x)^\E e$, with $(\,\cdot\,,\,\cdot\,)^\E$ the Euclidean inner product on $\Rl^d$, and is recalled in the following definition. As a shorthand, we write $\Rl_{>}^d:=e^\perp+\Rl_+\cdot e$ for the half space of $\Rl^d$ with positive $e$-coordinates.

\begin{definition}\label{def:states}
	Let $G^\E\subset E(d)$ and $G^\M\subset \PG(d)$ be subgroups such that $G^\M$ contains translations along the time direction $e$ and $G^\E$ is invariant under $g\mapsto r_e gr_e$. Let furthermore $(\E,\OOO,\alpha^\E)$ be a $G^\E$-covariant net on $\Rl^d$, and let $(\M,\OOO,\alpha^\M)$ be a $G^\M$-covariant net on $\Rl^d$.
	\begin{enumerate}
		\item A reflection positive functional on $\E$ is a continuous normalized linear functional $\sigma:\E\to\Cl$ such that
		\begin{itemize}
			\item $G^\E\ni g\longmapsto\sigma(A\alpha^\E_g(B))$ is continuous for all $A,B\in\E$,
			\item $\sigma\circ\alpha^\E_g=\sigma$ for all $g\in G^\E$,
			\item There exists an automorphism $\iota_e$ of $\E$ which acts covariantly, $\iota_e(\E(\OO))=\E(r_e\OO)$, $\OO\in\OOO$, such that $\iota_e\alpha^\E_g\iota_e=\alpha^\E_{r_egr_e}$ for all $g\in G^\E$, $\sigma\circ\iota_e=\sigma$, and
			\begin{align*}
					\sigma(\iota_e(A^*)A)\geq0
			\end{align*}
			for all $A\in\E_>:=\E(\Rl_>^d)$.
		\end{itemize}
		\item A vacuum state on $\M$ is a normalized, positive, linear functional $\om:\M\to\Cl$ such that
		\begin{itemize}
			\item $G^\M\ni g\longmapsto\om(A\alpha^\M_g(B))$ is continuous for all $A,B\in\M$,
			\item $\om\circ\alpha^\M_g = \om$ for all $g\in G^\M$,
			\item There is a weakly dense subset $\DD\subset\M$ such that $-i\left.\frac{d}{dt}\right|_{t=0}\om(A^*\alpha^\M_{t\cdot e,1}(A))$ exists and is non-negative for all $A\in\DD$.
		\end{itemize}
	\end{enumerate}
\end{definition}

Given an (undeformed) $E(d)$-covariant net $(\E,\OOO,\alpha^\E)$ with abelian $\E$ satisfying the time zero condition, and a reflection positive functional $\sigma$ on $\E$, a corresponding $\PG(d)$-covariant local net $(\M,\OOO,\alpha^\M)$ can be constructed \cite{Schlingemann:1999-1}. In the following, we will analyze this situation for the case of deformed field theories. Such models often have smaller symmetry groups than the full Euclidean respectively Poincar\'e groups\footnote{However, see \cite{DoplicherFredenhagenRoberts:1995, GrosseLechner:2007} for examples of fully covariant models.}. In fact, the basic commutation relations $[X_{\mu},X_{\nu}]=i\te_{\mu\nu}\cdot 1$ \eqref{eq:MoyalRelations} underlying Moyal space are invariant under all translations $X_\mu\mapsto X_\mu+x_\mu\cdot 1$, $x_\mu\in\Rl$, but only under those linear transformations $X_\mu\mapsto M_{\mu}^\nu X_\nu$, $M\in\GL(d)$, for which $M$ and $\te$ commute. We therefore introduce the (connected) {\em reduced Euclidean and Poincar\'e groups} as
\begin{align}
	E_\te(d)
	&:=
	\{(x,R)\,:\,(x,R)\in E(d)_0,\,\; R\te=\te R\}
	=
	\SO_\te(d)\ltimes\Rl^d
	\,,\\
	\PG_\te(d)
	&:=
	\{(x,\La)\,:\,(x,\La)\in \PG(d)_+^{\uparrow},\,\; \La\te=\te\La\}
	=
	\LG_\te(d)_+^\uparrow\ltimes\Rl^d
	\,,
	\label{eq:ReducedPoincareGroup}
\end{align}
and subsequently only consider models with such symmetry. In Section \ref{section:warping}, we will show how nets with reduced symmetry groups naturally appear as deformations of fully covariant models.

As a prerequisite for the analytic continuation between the Euclidean and Lorentzian setting, we now discuss the structure of $E_\te(d)$ and $\PG_\te(d)$ in more detail. As before, we fix a unit vector $e\in\Rl^d$ as a reference direction for reflection positivity in the Euclidean setting. After continuation to imaginary coordinates, this direction corresponds to the time direction in the Lorentzian setting, and since we want to study the case of space(times) with ``commutative time'', we assume that $e$ lies in the kernel of $\te$. Then $\te$ can be restricted to the hyperplane $e^\perp$, and we write $\vartheta:=\theta|_{e^\perp}$.

On $E_\te(d)$, the reflection $r_e:x\mapsto x-2(e,x)^\E e$ acts as an involutive automorphism $\gamma_e:~(x,R)\mapsto(r_e x,r_eRr_e)$. The subgroup $E_\te^e(d)\subset E_\te(d)$ of fixed points of $\gamma_e$ is isomorphic to $E_\vartheta(d-1)=\SO(d-1)_\vartheta\ltimes\Rl^{d-1}$, the reduced Euclidean group of $e^\perp\cong\Rl^{d-1}$ with noncommutativity $\vartheta$. Hence this involution induces a decomposition of the Lie algebra $\mathfrak{e}_\te(d)$ of $E_\te(d)$ into corresponding eigenspaces with eigenvalues $\pm1$,
\begin{align}
	\mathfrak{e}_\te(d)
	=
	\mathfrak{e}_\te^e(d) \oplus \mathfrak{m}_\te
	\,,
\end{align}
where $\mathfrak{e}_\te^e(d)=\mathfrak{e}_\vartheta(d-1)$ denotes the Lie algebra of $E_\te^e(d)\cong E_\vartheta(d-1)$. The pair $(\mathfrak{e}_\te(d),\mathfrak{e}_\te^e(d))$ has the structure of a symmetric Lie algebra,  {\em i.e.}, $[\mathfrak{e}_\te^e(d),\mathfrak{e}_\te^e(d)]\subset\mathfrak{e}_\te^e(d)$, $[\mathfrak{e}_\te^e(d),\mathfrak{m}_\te]\subset\mathfrak{m}_\te$, $[\mathfrak{m}_\te,\mathfrak{m}_\te]\subset\mathfrak{e}_\te^e(d)$ \cite{Helgason:1962}.

For a concrete description of these spaces, it is convenient to use orthonormal coordinates $(x_0,...,x_{d-1})$ of $\Rl^d$, with $e=(1,0,...,0)$. Let ${\sf p}_0,...,{\sf p}_{d-1}$ denote the corresponding generators of translations, and let ${\sf m}_{kl}$, $k<l$, $k,l=0,...,d-1$, denote the generators of rotations in the $x_k$-$x_l$-plane. Then $\mathfrak{e}_\te^e(d)$ is spanned by ${\sf p}_1,...,{\sf p}_{d-1}$ and all linear combinations of ${\sf m}_{kl}$, $k>0$, which commute with $\te$. The linear space $\mathfrak{m}_\te$ is spanned by ${\sf p}_0$ and all linear combinations of ${\sf m}_{0k}$, $k=1,...,d-1$, commuting with $\te$.

The dual symmetric Lie algebra of $\mathfrak{e}_\te(d)$ is defined as
\begin{align}
	\mathfrak{e}_\te(d)^*
	:=
	\mathfrak{e}_\te^e(d) \oplus i\mathfrak{m}_\te
	\,.
\end{align}
This is again a real Lie algebra, which is closely related to the reduced Poincar\'e group.

\begin{lemma}
	The connected, simply connected Lie group $E_\te(d)^*$ with Lie algebra $\mathfrak{e}_\te(d)^*$ is  the universal covering group of the reduced Poincar\'e group,
	\begin{align}
		E_\te(d)^*=\widetilde{\PG}_\te(d)\,.
	\end{align}
\end{lemma}
\begin{proof}
	For $\te=0$, this fact is well known \cite{FrohlichOsterwalderSeiler:1983, JorgensenOlafsson:1999}. Now $\mathfrak{e}_\te(d)^*\subset\mathfrak{e}_0(d)^*$ consists of all translation generators ${\sf p}_\mu$, $\mu=0,...,d-1$, and those elements of $\mathfrak{so}_0(d)^*$ which commute with $\te$, {\em i.e.}, $\mathfrak{e}_\te(d)^*$ coincides with the Lie algebra of the reduced Poincar\'e group $\PG_\te(d)$ \eqref{eq:ReducedPoincareGroup}. Hence $E_\te(d)^*$ is the unique connected simply connected Lie group with the same Lie algebra as $\PG_\te(d)$, that is, the universal covering group $\widetilde{\PG}_\te(d)$.
\end{proof}

\noindent{\em Remark:} This result shows that the same noncommutativity $\te$ can be used consistently for both the Euclidean and Lorentzian signature, and can intuitively be understood on the level of the Moyal commutation relations: When setting up Euclidean Moyal space via coordinates $X^\E_0,...,X^\E_{d-1}$ satisfying the relations $[X^\E_\mu,X^\E_\nu]=i\te^\E_{\mu\nu}$, and similarly Minkowski Moyal space via coordinates $X^\M_0,...,X^\M_{d-1}$ with $[X^\M_\mu,X^\M_\nu]=i\te^\M_{\mu\nu}$, one might expect that by some sort of Wick rotation, $X^\M_0=iX^\E_0$, whereas the spatial coordinates $X^\E_k=X^\M_k$, $k=1,...,d-1$, can be identified. In case of commuting time,  {\em i.e.} $\te^\E_{0\mu}=\te^\M_{0\mu}=0$, $\mu=0,...,d-1$, this reasoning would then imply coinciding noncommutativity parameters $\te^\E=\te^\M=:\te$ for both signatures.
\\
\\
The relevance of representing $E_\te(d)$ and $\widetilde{\PG}_\te(d)$ as dual Lie groups (in the sense defined above) lies in the fact that certain (virtual) representations of a Lie group and its dual are connected by analytic continuation \cite{FrohlichOsterwalderSeiler:1983}, and will be used subsequently.

%% file: symmetries.tex
\section{From $\boldsymbol{E_\te(d)}$-covariant nets on Euclidean space to\\ $\boldsymbol{\PG_\te(d)}$-covariant nets on Minkowski spacetime}\label{section:symmetries}

We now describe how to pass from a Euclidean $E_\te(d)$-covariant net on Euclidean space $\Rl^d$ to a $\PG_\te(d)$-covariant net on Minkowski spacetime $\Rl^d$. This construction will proceed in three steps: First, we consider an abstract $E_\te(d)$-covariant net $(\E,\OOO,\alpha^\E)$ and a reflection positive functional $\sigma$ on $\E$ w.r.t. some reference direction $e$ in the kernel of $\te$. This net can be represented on a Hilbert space $\Hil^\M$ by a procedure analogous to the GNS representation. On $\Hil^\M$, a virtual representation $V$ (involving semi groups of unbounded operators) of $E_\te(d)$ exists, and in the second step, we will construct a unitary representation $U$ of $E_\te(d)^*=\widetilde{\PG}_\te(d)$ via analytic continuation. This construction is well-known in the commutative case \cite{KleinLandau:1982,FrohlichOsterwalderSeiler:1983,Seiler:1982}, and we show that it carries over to the situation considered here. The third step consists in constructing a $\widetilde{\PG}_\te(d)$-covariant net $(\M,\OOO,\alpha^\M)$ on Minkowski spacetime by exploiting the $\widetilde{\PG}_\te(d)$-representation constructed before, and the time zero condition, and is based on the article \cite{Schlingemann:1999-1} discussing the analogous commutative situation.
\\
\\
So let $(\E,\OOO,\alpha^\E)$ be a $E_\te(d)$-covariant net of $C^*$-algebras on $\Rl^d$ and let $e$ be a unit vector with $\te e=0$. We assume that this net satisfies the time zero condition w.r.t. $e$, and that also the reflection $r_e$ inverting $e$ is represented by an automorphism $\iota_e$ on $\E$. Furthermore, let $\sigma$ be a $E_\te(d)$-invariant functional on $\E$ which is reflection positive w.r.t.~$e$.

By reflection positivity,
\begin{eqnarray}
	A,B \longmapsto \sigma \left( \iota_e({A}^*) {B} \right)
	\,,\qquad
	A,B\in\E_>\,,
\end{eqnarray}
defines a positive semi--definite sesquilinear form on $\E_>$. In particular, $\overline{\sigma( \iota_e({A}^*) {B})}=\sigma(\iota_e(B^*)A)$, and the Cauchy-Schwarz inequality holds. Dividing $\E_>$ by the null space $\NN_\sigma:=\{A\in\E_>\,:\,\sigma(\iota_e(A^*)A)=0\}$ of this sesquilinear form therefore yields a pre-Hilbert space  $\DD:=\E_>/\NN_\sigma$, with quotient map and scalar product denoted by $A\mapsto[ A ]_\sigma$ and
\begin{eqnarray}\label{skpr}
	\langle[{A}]_\sigma , [{B}]_\sigma \rangle := \sigma \left( \iota_e(A^*) {B} \right),
\end{eqnarray}
respectively. The vacuum Hilbert space $\Hil^\M$ is defined as the completion of $\DD$ in the norm given by \eqref{skpr}, and $\Om:=[1]_\sigma$ will play the role of the vacuum vector.

As the time reflection $\iota_e$ acts covariantly, the time zero algebra $\E_0=\E(e^\perp)=\iota_e(\E_0)$ is invariant under this automorphism. Thus for $A\in\E_0$, we also have $\iota_e(A)^*A\in\E_0\subset\E_>$, which implies that we have a well-defined GNS-type representation $\pi_\sigma$ of $\E_0$ on $\Hil^\M$,
\begin{align}\label{def:PiSigma}
	\pi_\sigma(A)[B]_\sigma := [AB]_\sigma
	\,,\qquad A\in\E_0,\,B\in\E_>
	\,.
\end{align}
This representation will be used later to generate the Minkowski net.

As a prerequisite for that step, however, we first need to discuss the action of the reduced symmetry groups $E_\te(d)$ and $\PG_\te(d)$ on $\Hil^\M$. To represent $E_\te(d)$, the basic idea is to consider a sufficiently small neighborhood $\UU$ of the identity in $E_\te(d)$ and define representing operators $V(g)$ by
\begin{align}\label{def:V}
	V(g)[A]_\sigma
	:=
	[\alpha_g^\E(A)]_\sigma
	\,,
	\qquad
	g\in\UU\,,
\end{align}
for suitable $A\in\E_>$. More precisely, for a given $g\in\UU$, we consider all regions $\OO\subset\Rl^d_>$ such that both, $g\OO$ and $\gamma_e(g)^{-1}g\OO$, are still contained in $\Rl^d_>$. For $A\in \E(\OO)$, the right hand side of \eqref{def:V} then exists because by covariance and isotony, $\alpha_g^\E(A)\in\E(g\OO)\subset\E_>$. Furthermore, the above assignment is well-defined,  {\em i.e.}, independent of the choice of representative in $[A]_\sigma$: In fact, for $A\in\NN_\sigma$ we can use the $E_\te(d)$-invariance of $\sigma$ to compute
\begin{align}
	\|[\alpha^\E_g(A)]_\sigma\|^2
	&=
	\sigma(\iota_e(\alpha_g^\E(A))^*\alpha^\E_g(A))
	\nonumber
	\\
	&=
	\sigma(\alpha_{\gamma_e(g)}^\E(\iota_e(A))^*\alpha^\E_g(A))
	\nonumber
	\\
	&=
	\sigma(\iota_e(A)^*\alpha^\E_{\gamma_e(g)^{-1}g}(A))
	\,.
	\label{eq:CauchySchwarz}
\end{align}
According to our assumption on the region $\OO$, we have $\gamma_e(g)^{-1}g\OO\subset\Rl^d_>$, and hence $\alpha^\E_{\gamma_e(g)^{-1}g}(A)\in\E_>$. Thus the Cauchy-Schwarz inequality yields $[\alpha^\E_g(A)]_\sigma=0$, which shows that $V(g)$ is well-defined. The subspace of $\Hil^\M$ which is spanned by all $[\E(\OO)]_\sigma$, where $\OO$ runs over the described set of regions, will be taken as the domain $\dom V(g)$ of $V(g)$.

For a common domain of definition, we choose a particular region $\OO$. Working in an orthonormal basis of $\Rl^d$ with coordinates $x=(x_0,...,x_{d-1})$ such that $e=(1,0,...,0)$, we consider the cone
\begin{align}
	C
	:=
	\{x\in\Rl^d\,:\,x_0>1+(x_1^2+...+x_{d-1}^2)^{1/2}\}
	\,.
\end{align}
Clearly there exists a neighborhood $\UU\subset E_\te(d)$ of the identity such that $gC$ and $\gamma_e(g)^{-1}gC$ are both contained in $\Rl^d_>$ for all $g\in\UU$. Furthermore, we can choose $\UU$ so large that it contains the full spatial subgroup $E_\te^e(d)$ because $gC\subset\Rl^d_>$ for all $g\in E_\te^e(d)$. The subspace
\begin{align}
	\DD_0
	:=
	[\E(C)]_\sigma
	\subset
	\Hil^\M
\end{align}
will be used as domain for our virtual representation.

\begin{proposition}\label{proposition:VirtualRepresentation}
	The data $(\UU,\DD_0,V)$ form a virtual representation \cite{FrohlichOsterwalderSeiler:1983} of $E_\te(d)$, i.e.,
	\begin{enumerate}
		\item $\DD_0\subset\Hil^\M$ is dense and for all $g\in\UU$, one has $\DD_0\subset\dom V(g)$.
		\item If $g_1,g_2$ and $g_1g_2$ lie in $\UU$ and $\Psi\in\DD_0$, then $V(g_2)\Psi\in\dom V(g_1)$ and
			\begin{align*}
				V(g_1)V(g_2)\Psi=V(g_1g_2)\Psi\,.
			\end{align*}
		\item For $g\in E_\te^e(d)$, the operator $V(g)$ extends to a unitary on all of $\Hil^\M$.
		\item For $g=\exp {\sf m}\in\UU$ with ${\sf m}\in\mathfrak{m}_\te$, the operator $V(g)$ is hermitian.
		\item Let $\Psi\in\DD_0$. Then $\UU \ni g\mapsto V(g)\Psi$ is strongly continuous at the identity.
		\item The translations in $e$-direction $\{V(\exp t{\sf p}_0)\}_{t\geq0}$ form a contraction semi-group.
	\end{enumerate}
\end{proposition}
\begin{proof}
	{\em b)} Let $g_1,g_2\in\UU$ such that also $g_1g_2\in\UU$. Then $g_1g_2C\subset\Rl^d_>$,  {\em i.e.} $V(g_2)\DD_0\subset\dom V(g_1)$. The group law $V(g_1)V(g_2)\Psi=V(g_1g_2)\Psi$, $\Psi\in\DD_0$, is clear from \eqref{def:V}.

	To establish {\em c)} and {\em d)}, we compute with $A,B\in\E(C)$ and $g\in\UU$
	\begin{align}\label{eq:hermitian}
		\langle V(g)[A]_\sigma,\,[B]_\sigma\rangle
		&=
		\sigma(\alpha_{r_e}^\E(\alpha_g^\E(A^*))B)
		=
		\sigma(\alpha^\E_{\gamma_e(g)}(\alpha_{r_e}^\E(A^*))B)
		\nonumber
		\\
		&=
		\sigma(\alpha_{r_e}^\E(A^*)\alpha^\E_{\gamma_e(g)^{-1}}(B))
		=
		\langle [A]_\sigma,\,V(\gamma_e(g)^{-1})[B]_\sigma\rangle
		\,,
	\end{align}
	yielding $V(g)^*\supset V(\gamma_e(g)^{-1})$. This shows that for a $\gamma_e$-invariant group element $g\in E_\te^e(d)$, we have $V(g)^*\supset V(g)^{-1}$, and once we have checked that $\DD_0$ is dense (part {\em a)}), it is clear that such $V(g)$ extend to unitaries on all of $\Hil^\M$. On the other hand, for $g=\exp {\sf m}\in\UU$, ${\sf m}\in\mathfrak{m}_\te$, we have $\gamma_e(g)=g^{-1}$, and thus the representing operator $V(g)$ is hermitian.

	{\em e)} Let $g\in\UU$ and $A\in\E(C)$. Then we compute as in \eqref{eq:CauchySchwarz}
	\begin{align*}
		\|V(g)[A]_\sigma-[A]_\sigma\|^2
		&=
		\sigma
			\left(
				\iota_e(A^*)
				\left(
				\alpha^\E_{\gamma_e(g)^{-1}g}(A)
				-
				\alpha^\E_{\gamma_e(g)^{-1}}(A)
				-
				\alpha^\E_g(A)
				+
				A
				\right)
			\right)
			\,.
	\end{align*}
	In the limit where $g$ approaches the identity in $E_\te(d)$, this norm difference vanishes because of the assumed continuity of $\sigma$ (Definition \ref{def:states} b)).

	{\em f)} For positive parameters $t\geq0$, the domain of the $e$-translations $V_1(t):=V(\exp t{\sf p}_0)$ is the dense subspace $\dom V_1(t)=\DD=\E_>/\NN_\sigma$. As the generator ${\sf p}_0$ of translations along $e$ lies in the $-1$ eigenspace $\mathfrak{m}_\te$ of $\gamma_e$, we find by application of {\em d)} that $V_1(t)$ is a hermitian operator for any $t\geq0$. As $\Rl_+\ni t\mapsto V_1(t)$ is strongly continuous on $\DD$ by the continuity of $\sigma$, it follows that $V_1$ is a symmetric local semi group \cite{KleinLandau:1981}. In particular, there exists a self-adjoint operator $H$ with $\DD\subset {\rm dom}\,e^{-tH}$ for any $t\geq0$, such that $V_1(t)=e^{-tH}$ on $\DD$. For $A\in\E_>$, $t\geq0$, we estimate
	\begin{align*}
		\|e^{-tH}[A]_\sigma\|
		=
		\|V_1(t)[A]_\sigma\|
		=
		\|[\alpha^\E_{\exp t{\sf p}_0}(A)]_\sigma\|
		\leq
		\|\alpha^\E_{\exp t{\sf p}_0 }(A)\|_\E
		=
		\|A\|_\E
		\,.
	\end{align*}
	Since this bound is uniform over all $t\geq0$, and $\DD\subset\Hil^\M$ is dense, it follows that $H$ is positive,  {\em i.e.}, $V_1$ is a contraction semi group.

	{\em a)} With this information, we can now prove that $\DD_0$ is a dense subspace of $\Hil^\M$ as well, using a Reeh-Schlieder type argument (see also \cite{Schlingemann:1999-1}).  For $\Psi\in \DD_0^\perp$ and $A\in\E(\OO)$ with some bounded $\OO\subset\Rl^d_>$, consider the function $f:\Rl\to\Cl$,
	\begin{eqnarray}
		f(s)
		:=
		\langle \Psi , e^{-sH}[A]_\sigma \rangle
		=
		\langle \Psi ,  [\alpha^\E_{\exp s{\sf p}_0}(A) ]_\sigma \rangle
		\,.
	\end{eqnarray}
	Since $H$ is positive and $V_1$ is strongly continuous, $f$ extends to a holomorphic function in the right half plane, with continuous boundary values. Moreover, as $\OO$ is bounded, there exists $s_0>0$ such that $\OO+s\cdot e\subset C$ and thus $\alpha^\E_{\exp s{\sf p}_0}(A)\in\E(C)$ for all $s\geq s_0$. Hence $f(s)=0$ for $s\geq s_0$, which by the analyticity of $f$ implies $0=f(0)= \langle \Psi , [A]_\sigma \rangle$, {\em i.e.}, $\Psi\perp[\E(\OO)]_\sigma$. But by definition of $\E_>$, the union of all $\E(\OO)$, where $\OO$ runs over all bounded regions in $\Rl^d_>$, is norm-dense in $\E_>$, and by construction of the Hilbert space, $[\E_>]_\sigma$ is a dense subspace of $\Hil^\M$. Hence $\Psi=0$, which proves the density of $\DD_0\subset\Hil^\M$. The fact that $\DD_0\subset\dom V(g)$ for all $g\in\UU$ is clear from the definitions of $\DD_0$ and the domains ${\rm dom}V(g)$.
\end{proof}

We are thus in the situation of a virtual representation of a class 2 symmetric space according to the terminology used in \cite[Thm.~3]{FrohlichOsterwalderSeiler:1983}. In that article it was shown that under the conditions verified in Proposition \ref{proposition:VirtualRepresentation}, $V$ can be analytically continued to a unitary representation $\tilde{U}$ of $E_\te(d)^*=\widetilde{\PG}_\te(d)$. In the concrete situation at hand, $\tilde{U}$ actually descends to a unitary representation $U$ of the reduced Poincar\'e group $\PG_\te(d)$ itself: For $\te=0$, this follows from the analysis in \cite{KleinLandau:1982}, where the analytic continuation of $V$ was carried out for $E_0(d)$ and shown to result in a unitary representation of $\PG_0(d)$ instead of the universal covering group. This feature then restricts to the reduced group $\PG_\te(d)\subset\PG_0(d)$ for $\te\neq0$.

Furthermore, the vacuum vector $\Om=[1]_\sigma\in\DD_0$ is invariant under all $V(g)$, $g\in\UU$, as can be seen from \eqref{def:V}. As $U$ is obtained from $V$ by analytic continuation, $\Om$ is invariant under the representation $U$ as well.

We summarize these observations in the following theorem.

\begin{theorem}
	There exists a strongly continuous unitary representation $U$ of $\PG_\te(d)$ on $\Hil^\M$ such that
	\begin{enumerate}
		\item $U(g)=V(g)$ for $g\in E_\te^e(d)\subset \PG_\te(d)$.
		\item $U(g)\Om=\Om$ for all $g\in\PG_\te(d)$.
	\end{enumerate}
\end{theorem}

The representation $U$ provides us in particular with a strongly continuous representation $x\mapsto U(x,1)$ of the translation group $\Rl^d$, generated by $P_0:=H$ and the Euclidean momentum operators $P_1,...,P_{d-1}$. The joint spectrum of these $d$ commuting selfadjoint operators can be characterized as follows.

\begin{proposition}
	The joint spectrum $S$ of the generators $P_0,...,P_{d-1}$ of the translations $U(x,1)=e^{iP\cdot x}$ is a $\PG_\te(d)$-invariant subset of $\Rl^d$ satisfying $\{p_0 : (p_0,...,p_{d-1})\in S\}\subset\Rl_+$ and $0\in S$.
\end{proposition}
\begin{proof}
	As $U(x,1)$ extends to a representation of ${\PG}_\te(d)$, the joint spectrum of its generators is a $\PG_\te(d)$-invariant subset of $\Rl^d$. Moreover, $P_0$ is positive by Proposition~\ref{proposition:VirtualRepresentation}~{\em f)}, and hence $\{p_0 : (p_0,...,p_{d-1})\in S\}={\rm spec} P_0\subset\Rl_+$. We have $0\in S$ because $\Om$ is a invariant under translations.
\end{proof}

\noindent{\em Remark:} For $\te=0$, the above restrictions on the spectrum imply the well-known spectrum condition, stating that $S$ is a subset of the closed forward lightcone. However, depending on the space-time dimension $d$ and the form of the noncommutativity $\te\neq0$, the shape of $S$ is less restricted in general. For example, if $d>2$ is even and $\ker\te=\{(p_0,p_1,0,...,0)\,:\,p_0,p_1\in\Rl\}$, the boosts $\La_1(\beta)$ in $x_1$-direction lie in\footnote{In fact, this is also the case for suitable $\te$ of full rank.} $\PG_\te(d)$, and it is easy to see that $(\La_1(\beta)p)_0=\cosh(\beta)p_0+\sinh(\beta)p_1$ is negative for certain $\beta$ if $p_0<|p_1|$. Hence in this case, we only get $S\subset Y$, where $Y$ denotes the closed {\em lightwedge}
\begin{align*}
	Y
	:=
	\{p\in\Rl^d\,:\,p_0\geq|p_1|\}
	\,.
\end{align*}
Such a form of the spectrum has also been considered in other discussions of quantum field theory on non-commutative spacetimes \cite{AlvarezGaumeVazquezMozo:2003}.
\\
\\
Having constructed a suitable Hilbert space $\Hil^\M$ and representation $U$ of the reduced Poincar\'e group on it, we now proceed to step three of our construction and discuss how to define a ${\PG}_\te(d)$-covariant net $(\M,\OOO,\alpha^\M)$ on $\Hil^\M$, where
\begin{align}
	\alpha^\M_g(A)
	:=
	U(g)AU(g)^{-1}
	\,,\qquad
	g\in {\PG}_\te(d)
	\,,
\end{align}
denotes the adjoint action of $U$ by automorphisms of $\B(\Hil^\M)$, and $\OOO$ is the family of all open bounded regions in $\Rl^d$.

We first recall the representation $\pi_\sigma$ \eqref{def:PiSigma} of the time zero algebra $\E_0=\E(e^\perp)$ on $\Hil^\M$. By the time zero condition on the net $\E$, this algebra is in particular non-trivial. We now define for any open bounded region $\OO\subset\Rl^d$
\begin{align}\label{def:net-M}
	\M(\OO)
	:=
	\bigvee_{K \subset e^\perp,\,g\in\PG_\te(d)\atop gK\subset \OO}
		\alpha^\M_g(\pi_\sigma(\E_0(K))
\end{align}
as the smallest von Neumann algebra in $\B(\Hil^\M)$ containing the operators $\alpha^\M_g(\pi_\sigma(A))$ for all $A\in\E_0(K)$, $K\subset e^\perp$, $g\in \PG_\te(d)$ such that $gK\subset \OO$. Denoting by $\OOO$ the family of all open bounded regions in $\Rl^d$, we find the following main result of this section.

\begin{theorem}
	\begin{enumerate}
		\item The above constructed algebras $\M(\OO)$ \eqref{def:net-M} and action $\alpha^\M$ form a $\mathcal P_\theta(d)$-covariant net $(\M,\OOO,\alpha^\M)$ of von Neumann algebras on Minkowski spacetime which satisfies the time zero condition.
		\item The state $\om(A):=\langle\Om,A\Om\rangle$ is a $\PG_\te(d)$-invariant vacuum state on $\M$.
	\end{enumerate}
\end{theorem}
\begin{proof}
{\em a)} By definition, the $\M(\OO)$ are von Neumann algebras, and $\alpha^\M={\rm Ad}\,U$ is an automorphic $\PG_\te(d)$-action on $\B(\Hil^\M)$. According to Definition \ref{def:net}, we have to show that isotony and covariance holds.

Regarding the first property, let $\OO_1 \subset \OO_2$ be an inclusion of two regions in $\Rl^d$, and $K\subset e^\perp$ such that there exists $g\in\PG_\te(d)$ with $gK\subset\OO_1$. Then clearly also $gK\subset\OO_2$, and hence any $\alpha^\M_g(\pi_\sigma(A))\in\M(\OO_1)$, $A\in\E_0(K)$, is contained in $\M(\OO_2)$ as well. As these operators generate $\M(\OO_1)$, isotony holds, {\em i.e.}, $\M(\OO_1)\subset\M(\OO_2)$.

Regarding covariance under $\alpha^\M$, let $\OO\subset\Rl^d$ be open, $K\subset e^\perp$, $A\in\E_0(K)$, and $g\in\PG_\te(d)$ such that $gK\subset\OO$. Then $B:=\alpha^\M_g(\pi_\sigma(A))\in\M(\OO)$, and for any $h\in\PG_\te(d)$, we have $\alpha^\M_h(B)=\alpha^\M_{hg}(\pi_\sigma(A))$. Since clearly $hg K\subset h\OO$, we have $\alpha^\M_h(B)\in\M(h\OO)$, and as $\M(\OO)$ is generated by operators of this form, also  $\alpha^\M_h(\M(\OO))\subset\M(h\OO)$. By also considering the inverse transformation $h^{-1}$, we arrive at the covariance property $\alpha^\M_h(\M(\OO))=\M(h\OO)$.

The time zero condition is satisfied by the very construction of the net $\M$.

{\em b)} Since $\Om$ is a $U$-invariant unit vector, it is clear that $\om$ is an $\alpha^\M$-invariant state on $\B(\Hil^\M)$ and thus on $\M\subset\B(\Hil^\M)$, too. Furthermore, as $U$ is strongly continuous,
\begin{align*}
	\PG_\te(d) \ni g \longmapsto\om(A\alpha^\M_g(B))=\langle A^*\Om,\,U(g)B\Om\rangle
\end{align*}
is continuous for all $A,B\in\M$, as required in Definition \ref{def:states}.

Finally, as $\M$ is weakly closed and invariant under the translations $\alpha^\M_{x,1}$, $x\in\Rl^d$, and $U$ is strongly continuous, $\M$ contains a strongly dense subalgebra $\M^\infty$ consisting of operators $A$ for which $x\mapsto\alpha^\M_{x,1}(A)$ is smooth. In particular, the function $\Rl\ni t\mapsto\om(A^*\alpha^\M_{t\cdot e,1}(A))$ is differentiable at $t=0$ for such $A$, and $-i\frac{d}{dt}|_{t=0}\om(A^*\alpha^\M_{t\cdot e,1}(A))=\langle A\Om,P_0A\Om\rangle\geq0$ because $P_0$ is positive. This completes the list of properties of $\om$ required in Definition~\ref{def:states}.
\end{proof}

%% file: warping.tex
\section{Wick rotation of deformed nets}\label{section:warping}

So far we have discussed how to relate a net on Euclidean space $\Rl^d$ to a net on Minkowski spacetime by some kind of Wick rotation. Motivated by the commutation relations defining Moyal space(-time), these nets were assumed to transform covariantly under the reduced symmetry groups $E_\te(d)$ and $\PG_\te(d)$, respectively. In the present section, we will show how such nets arise naturally in the context of field theories on noncommutative spaces. As we are aiming for a general, largely model-independent connection between Euclidean and Lorentzian theories, we work here with a formulation of field theories on noncommutative spaces which can be applied to a wide range of models. The basic idea is to consider field theories on ordinary commutative $\Rl^d$, and describe their counterparts on Moyal space by an appropriate deformation procedure. Such a deformation incorporates the noncommutative nature of spacetime into the relations of the observable algebra, and yields an effective description of field theory on Moyal space(-time).

There exist two closely related versions of this deformation. The first version consists in deforming the product in the observable algebra to the Rieffel-product $\times_\te$ \eqref{eq:DeformedProduct} \cite{Rieffel:1992}, a generalization of the Weyl-Moyal star product recalled below. This is a representation independent deformation and thus adequate for transferring a Euclidean field theory, given by an abstract net of $C^*$-algebras, to Moyal space. The second version of this deformation is carried out in a representation of the observable algebra on a Hilbert space, and can be formulated by deforming the representing operators $A\mapsto A_\te$, but keeping their product as the usual operator product unchanged. This latter procedure, known as warped convolution \cite{BuchholzSummers:2008}, will be the adequate tool to deform field theories given in concrete representations, such as the Minkowski space nets ({\em cf.} Section \ref{section:symmetries}).

The relations between Rieffel's product deformation and warped convolution have been clarified in \cite{BuchholzLechnerSummers:2010}, see also \cite{LechnerWaldmann:2011} for a unified description. Furthermore, both deformations can be formulated in both Euclidean and Lorentzian signature, which makes them applicable to the situation at hand.

For the sake of self-containedness, we recall here briefly the motivation for the warped convolution deformation from the point of view of field theories on noncommutative Minkowski spacetime \cite{GrosseLechner:2007, GrosseLechner:2008}. Starting from a quantum field operator $\phi$ in a vacuum representation on a Hilbert space $\Hil$, and a representation of the Moyal coordinates $X_\mu$ with noncommutativity $\te$ on a Hilbert space $\K$, one considers the ``noncommutative version'' $\phihat_\te$ of $\phi$, formally defined as, $x\in\Rl^d$,
\begin{align*}
	\hat{\phi}_\te(x)
	=
	(2\pi)^{-d/2}\int dp\,
	e^{ip_\mu(X^\mu+x^\mu\cdot 1)}
	\otimes
	\phiti(p)
	\,.
\end{align*}
Passing to a different vacuum representation of the field $\hat{\phi}_\te$ on $\Hil$ then yields the warped convolution $\phi_\te$ of $\phi$, recalled below.

As we will demonstrate, Rieffel deformation and warped convolution deform $E_0(d)$- respectively $\PG_0(d)$-covariant nets on $\Rl^d$ to nets which are covariant under the reduced symmetry groups $E_\te(d)$ respectively $\PG_\te(d)$, and thus produce models of the form considered in the previous section. Our main result is, roughly speaking, that Wick rotation and deformation commute. More precisely, when starting from an abstract $E_0(d)$-covariant Euclidean net $\E$, we can consider its Rieffel deformation $\E_\te$, a $E_\te(d)$-covariant net, and the Wick rotation $\M_\te$ thereof, a $\PG_\te(d)$-covariant net. On the other hand, we can also first Wick rotate $\E$ to a $\PG_0(d)$-covariant net $\M$ by the Wick rotation of Section~\ref{section:symmetries} with $\te=0$ \cite{Schlingemann:1999-1}, and then deform by warped convolution to a $\PG_\te(d)$-covariant net $\widetilde{\M}_\te$. We will show that there exists a net isomorphism $\varphi$ such that the diagram
\begin{diagram}\label{diagram}
	&& \E && \\
	& \ldTo^{\text{\small deformation}} && \rdTo^{\text{\small Wick rotation}} & \\
	\E_\te &&&& \M \\
	\dTo^{\text{\small Wick rotation}} &&&& \dTo_{\text{\small deformation}} \\
	\M_\te &&\rDoubleArrow_\varphi && \widetilde{\M_\te}
\end{diagram}
commutes.
\\
\\
Before entering this analysis, we recall some facts about Rieffel deformations \cite{Rieffel:1992} and warped convolutions \cite{BuchholzLechnerSummers:2010}. In the context of Rieffel deformations, one considers a $C^*$-algebra $\A$ with a strongly continuous $\Rl^n$-action by automorphisms, a non-degenerate bilinear form $(\,\cdot\,,\,\cdot\,)$ with determinant $\pm1$ on $\Rl^n$, and a linear map $\te$ on $\Rl^n$ which is antisymmetric w.r.t. this bilinear form. Between elements $A\in\A$ for which $x\mapsto\alpha_x(A)$ is smooth, the {\em Rieffel product} is defined as
\begin{align}\label{eq:RieffelProduct}
	A\times_\te B
	:=
	(2\pi)^{-n}\int dp\,dx\,e^{i(p,x)}\alpha_{\te p}(A)\alpha_x(B)
	\,.
\end{align}
Understood in an oscillatory sense, this integral converges to a smooth element in $\A$. The integration runs here over $\Rl^n\times\Rl^n$, but it is not difficult to show that one can actually restrict to the image ${\rm Im}\te$ of $\te$,
\begin{align}\label{eq:RieffelProductReduced}
	A\times_\te B
	=
	(2\pi)^{-(n-\dim\ker\te)}\int_{{\rm Im}\te\times {\rm Im}\te} dp\,dx\,e^{i(p,x)}\alpha_{\te p}(A)\alpha_x(B)
	\,.
\end{align}
Since only translations along ${\rm Im}\te$ enter into this integral, we reserve the symbol $\A^\infty$ for the dense subalgebra of all elements $A\in\A$ such that ${\rm Im}\te\ni x\mapsto\alpha_x(A)$ is smooth. The product $A\times_\te B$ is then well-defined for $A,B\in\A^\infty$. We recall some facts about this product from Rieffel's work \cite{Rieffel:1992}:
\begin{lemma}\label{lemma:RieffelProduct}
	The Rieffel product has the following properties:
	\begin{enumerate}
		\item $\times_\te$ is an associative product on $\A^\infty$.
		\item For $A,B\in\A^\infty$, one has
		\begin{align*}
			(A\times_\te B)^*
			=
			B^*\times_\te A^*
			\,,\\
			A\times_\te 1=A=1\times_\te A
			\,.
		\end{align*}
		\item There exists a $C^*$-norm $\|\cdot\|_\te$ on $(\A^\infty, \times_\te)$.
	\end{enumerate}
\end{lemma}

The ${}^*$-algebra given by the linear space $\A^\infty$ equipped with the product $\times_\te$ is denoted  $\A_\te^\infty$, and the $C^*$-algebra obtained by completing $\A_\te^\infty$ in the norm $\|\cdot\|_\te$ is denoted $\A_\te$. It is known that the translations $\alpha_x|_{\A^\infty}$ extend to a strongly continuous automorphic $\Rl^n$-action $\alpha^\te$ of $\A_\te$ \cite{Rieffel:1992}. We will need the following two further statements about extending data from $\A_\te^\infty$ to $\A_\te$.
\begin{lemma}\label{lemma:RieffelProduct2}
	\begin{enumerate}
		\item Let $\A_1, \A_2\subset\A$ be $C^*$-subalgebras which are invariant under $\alpha$, and let $\beta:\A_1\to\A_2$ be an isomorphism such that $\beta\circ\alpha_x=\alpha_{Mx}\circ\beta$ for all $x\in\Rl^n$ and some $M\in{\rm GL}(n)$ with $M^T=M^{-1}$ and $M\te=\te M$, where the transpose $M^T$ refers to the bilinear form used in \eqref{eq:RieffelProduct}. Then $\beta(\A_1^\infty)=\A_2^\infty$, and $\beta|_{\A_1^\infty}$ extends to an isomorphism $\beta^\te:\A_{1,\te}\to\A_{2,\te}$ of the deformed $C^**$-algebras $\A_{1,\te}$, $\A_{2,\te}$, such that $\beta^\te\circ \alpha^\te_x=\alpha^\te_{Mx}\circ \beta^\te$ for all $x\in\Rl^n$.
		\item Let $\nu$ be an $\alpha$-invariant linear continuous functional on $\A$. Then $\nu|_{\A^\infty}$ extends to a linear continuous functional $\nu^\te$ on $\A_\te$, and there holds
		\begin{align}\label{eq:nuAB}
			\nu(A\times_\te B)
			=
			\nu(AB)
			\,,\qquad
			A,B\in\A^\infty.
		\end{align}
	\end{enumerate}
\end{lemma}
\begin{proof}
	{\em a)} This is a combination of Thm. 5.12 and Prop. 2.12 of \cite{Rieffel:1992}.

	In case $\nu$ is positive, part {\em b)} of this lemma was proven in \cite[Thm.~4.1]{Rieffel:1993}. In the more general case relevant here, we can decompose $\nu$ into a linear combination of four positive linear functionals $\nu_1,...,\nu_4$ on $\A$ (Jordan decomposition). As positive functionals, these are in particular continuous and moreover, this decomposition preserves $\alpha$-invariance. Hence by Rieffel's result, we have continuous linear functionals $\nu_k^\te$, $k=1,..,4$, on $\A_\te$, which satisfy \eqref{eq:nuAB} on the smooth subalgebra. Taking the same linear combination of the $\nu^\te_k$ as before then yields the desired functional $\nu^\te$.
\end{proof}
\\
\\
We now come to the summary of the structure of warped convolutions, the second version of the deformation procedure. In this setting, one considers a Hilbert space $\Hil$ carrying a weakly continuous unitary representation $U$ of $\Rl^n$, and a bilinear form $(\,\cdot\,,\,\cdot\,)$ and matrix $\te$ as above. The representation defines a dense subspace in $\Hil$ of smooth vectors and a weakly dense ${}^*$-algebra in $\B(\Hil)$ of smooth operators. For smooth $A$ and $\Psi$, one defines
\begin{align}\label{eq:Warping}
	A_\te\Psi
	:=
	(2\pi)^{-n}\int dp\,dx\,e^{i(p,x)}\,U(\te p)AU(\te p)^{-1}U(x)\Psi
	\,,
\end{align}
where the integral has to be understood in a strong oscillatory sense. One can show that this assignment yields an operator $A_\te$ which can be continuously extended to all of $\Hil$. In complete analogy to \eqref{eq:RieffelProductReduced}, one has
\begin{align}\label{eq:WarpingReduced}
	A_\te \Psi
	=
	(2\pi)^{-(n-\dim\ker\te)}\int_{{\rm Im}\te\times {\rm Im}\te} dp\,dx\,e^{i(p,x)}\,U(\te p)AU(\te p)^{-1}U(x)\Psi
	\,.
\end{align}
As only the translations along ${\rm Im}\te$ enter here, we will denote the subspace of all vectors $\Psi\in\Hil$ such that ${\rm Im}\te\ni x\mapsto U(x)\Psi$ is smooth by $\DD\subset\Hil$, and the subalgebra of all operators such that ${\rm Im}\te\ni x\mapsto U(x)AU(x)^{-1}$ is smooth (in norm) by $\C^\infty\subset\B(\Hil)$.
Some relevant properties of the deformation map $A\to A_\te$ are summarized below.
\begin{lemma}{\em\bf \cite{BuchholzLechnerSummers:2010}}\label{lemma:Warping}
	Let $A,B\in\C^\infty$. Then
	\begin{enumerate}
		\item $(A_\te)^*=(A^*)_\te$ and $1_\te=1$.
		\item Let $V$ be a unitary operator on $\Hil$ such that $VU(x)V^*=U(Mx)$, $x\in\Rl^n$, for some $M\in\GL(n)$ with $M^T=M^{-1}$, $M\te=\te M$, where the transpose $M^T$ refers to the bilinear form used in \eqref{eq:Warping}. Then $VA_\te V^*=(VAV^*)_{\te}$. In particular, $U(x)A_\te U(-x)=(U(x)AU(-x))_\te$ for all $x\in\Rl^n$.
		\item Let $\Om\in\Hil$ be a $U$-invariant vector. Then $A_\te\Om=A\Om$.
	\end{enumerate}
\end{lemma}

Regarding the connection between the Rieffel product and warped convolution, we will write $A\times_\te B$ for the Rieffel product \eqref{eq:RieffelProduct} of two smooth operators $A,B\in\C^\infty\subset\B(\Hil)$ w.r.t. the action $\alpha_x(A):=U(x)AU(x)^{-1}$.
\begin{lemma}{\em\bf \cite{BuchholzLechnerSummers:2010}}\label{lemma:WarpingAndRieffel}
	\begin{enumerate}
		\item Let $A,B\in\C^\infty$. Then $A\times_\te B\in\C^\infty$, and $A_\te B_\te = (A\times_\te B)_\te$.
		\item Let $\A$ be a $C^*$-algebra with strongly continuous $\Rl^n$-action $\alpha$ and $\pi$ an $\alpha$-covariant representation of $\A$ on $\Hil$, i.e., $U(x)\pi(A)U(x)^{-1}=\pi(\alpha_x(A))$, $A\in\A$. Then $\pi(\A^\infty)\subset\C^\infty$, and the map $\pi_\te(A):=\pi(A)_\te$, $A\in\A^\infty$, extends continuously to an $\alpha$-covariant representation of the deformed $C^*$-algebra $\A_\te$.
	\end{enumerate}
\end{lemma}

Having recalled the mathematical structure, we will now step by step discuss the various arrows in the diagram on page \pageref{diagram}. Our initial data consist of a $E_0(d)$-covariant net $(\E,\OOO,\alpha^\E)$ on $\Rl^d$, with $\OOO$ the family of open subsets of $\Rl^d$, and $\alpha^\E$ strongly continuous\footnote{In case that $\alpha^\E$ is not strongly continuous, we could pass to a representation of the Euclidean net on a Euclidean Hilbert space $\Hil^\E$ as indicated in Section~\ref{section:nets}. Making use of the warped convolution setting outlined above, we could then work with a unitary representation $U$ of $\Rl^d$ on $\Hil^\E$ and deformations of operators in von Neumann algebras. But for technical convenience, we stick to the assumption of strongly continuous $\alpha^\E$ here.}. Furthermore, $e\in\Rl^d$ is a unit vector, $\iota_e$ an automorphism of $\E$ which implements the reflection $r_e$, and $\sigma$ a corresponding reflection positive functional on $\E$. We require that $\E$ satisfies the time zero condition, and denote its time zero $C^*$-algebras by $\E_0(K)$, $K\subset e^\perp$. Finally, $\te$ is a $(d\times d)$-matrix which is antisymmetric w.r.t. the Euclidean inner product on $\Rl^d$ and satisfies $\te e=0$.

Starting from these data, we now describe as the first step the Rieffel-type deformation of the net $\E$ to some Euclidean net $\E_\te$ on Moyal space. On the level of the global algebra $\E$, we have the replacement $\E\to\E_\te$, where $\E_\te$ denotes the deformed $C^*$-algebra with product $\times_\te$ \eqref{eq:RieffelProduct}, defined with the Euclidean inner product $(\,\cdot\,,\,\cdot\,)^\E$ in the oscillating phase, the $\Rl^d$-action $\alpha^\E|_{\Rl^d}$, and the norm $\|\cdot\|_\te$. According to Lemma \ref{lemma:RieffelProduct2} {\em a)} with $\A_1=\A_2=\E$ and $\alpha_1=\alpha_2=\alpha^\E|_{\Rl^d}$, the deformed $C^*$-algebra $\E_\te$ carries an automorphic action $\alpha^{\E,\te}$ of $E_\te(d)$ which coincides with $\alpha^\E|_{E_\te(d)}$ on the smooth subalgebra $\E^\infty$. Furthermore, $\iota_e$ gives rise to an automorphism $\iota_e^\te$ representing the reflection $r_e$ on $\E_\te$.

Concerning the net structure of $\E_\te$, it is clear from formula \eqref{eq:RieffelProductReduced} that for general regions $\OO\in\OOO$, the set $\E(\OO)^\infty$ will not be an algebra for the product $\times_\te$ as this involves integration over all translations along the ``noncommutative directions'' in ${\rm Im}\te$. Making use of the orthogonal decomposition $\Rl^d=\ker\te\oplus{\rm Im}\te$, we are therefore led to consider the set of cylindrical regions
\begin{align}\label{eq:Cylinders}
	\Zyl_\te
	:=
	\{\OO\oplus{\rm Im}\te :\, \OO\subset\ker\te\;\text{open and bounded}\}
	\,.
\end{align}
Cylinders $Z\in\Zyl_\te$ are clearly translationally invariant in the noncommutative directions, {\em i.e.}, $Z+x=Z$ for all $x\in{\rm Im}\te$. As a consequence, the translation automorphisms $\alpha^\E_x$, $x\in{\rm Im}\te$, restrict to $\E(Z)$, and hence we have the Rieffel-deformed cylinder algebras $\E_\te(Z)$.

As linear spaces, $\E_\te(Z)^\infty=\E(Z)^\infty$, and thus $\E_\te(Z_1)^\infty\subset\E_\te(Z_2)^\infty$ for cylinders $Z_1\subset Z_2$, $Z_1,Z_2\in\Zyl_\te$. This inclusion remains valid after closing in $\|\cdot\|_\te$, yielding a net $Z\mapsto\E_\te(Z)$ over cylinder regions. Furthermore, it is clear from the definition \eqref{eq:Cylinders} that the family $\Zyl_\te$ is invariant under the reduced Euclidean group $E_\te(d)$. Making use of Lemma \ref{lemma:RieffelProduct2} {\em a)} with $\A_1=\E(Z)$, $\A_2=\E(gZ)$, $g\in E_\te(d)$, and $\alpha$ the action of the translations along ${\rm Im}\te$ on these algebras, we see that
\begin{align}
	\alpha^{\E,\te}_g(\E_\te(Z))
	=
	\E_\te(gZ)
	\,,\qquad
	g\in E_\te(d),\;Z\in\Zyl_\te\,.
\end{align}
We summarize these findings in the following proposition.
\begin{proposition}
	The data $\E_\te(Z)$, $Z\in\Zyl_\te$, and $\alpha^{\E,\te}$ constructed above form a $E_\te(d)$-covariant net $(\E_\te,\Zyl_\te,\alpha^{\E,\te})$ of $C^*$-algebras.{\hfill $\square$}
\end{proposition}

This step completes the deformation $\E\to\E_\te$. Next we want to proceed to a Minkowski space version $\M_\te$ of $\E_\te$. As explained in Section \ref{section:symmetries}, two further properties of $(\E_\te,\Zyl_\te,\alpha^{\E,\te})$ are necessary for this, a reflection positive functional and a time zero condition.

Regarding the time zero condition, we introduce a family of time zero slices analogous to the cylinders \eqref{eq:Cylinders},
\begin{align}
	\Slic_\te
	:=
	\{\{0\}\oplus K\oplus {\rm Im}\te  :\, K\subset\ker\te,\,K\perp e,\; K \;\text{open and bounded}\}
	\,,
\end{align}
where the direct sum refers to the orthogonal split $\Rl^d=\Rl e\oplus(e^\perp\cap\ker\te)\oplus{\rm Im}\te$. These slices are the time zero components of the cylinders in $\Zyl_\te$, that is, $\Slic_\te=\{Z\cap e^\perp\,:\,Z\in\Zyl_\te\}$. We need the following lemma, which for later use we state also in a Lorentzian version. As in Section \ref{section:symmetries}, we assume for the Lorentzian version that $e$ denotes the time direction, {\em i.e.}, the Minkowski metric is $\eta=-1\oplus1$ on $e\oplus e^\perp$, and in particular commutes with $\te$ since $\te$ is antisymmetric and $\te e=0$.

\begin{lemma}\label{lemma:SlicesAndCylinders}
	Let $S\in \mathsf \Slic_\te$ and $g\in E_0(d)$ (respectively $g\in\PG_0(d))$ such that $gS\subset Z$ for some $Z\in \Zyl_\te$. Then there exist $g_1\in E_\te(d)$, $g_2\in E_0(d)$ (respectively $g_1\in \PG_\te(d)$, $g_2\in \PG_0(d)$) such that $g=g_1g_2$ and $g_2S=S$.
\end{lemma}
\begin{proof}
	We write $g=(x,M)$ referring to the semidirect product structure of $E_0(d)=\Rl^d\rtimes {\rm SO}(d)$ respectively $\PG_0(d)=\Rl^d\rtimes \LGpo(d)$. The set $MS$ satisfies $MS+x=MS$ for $x\in M\rm Im \te$. If $M\rm Im\te \not\subset\rm Im\te$, it follows that $MS$ is not bounded in projection to $\ker \te$,  and can thus not be contained in an element of the family $\Zyl_\te$. Thus the assumption $gS\subset Z$ implies $M{\rm Im}\te\subset{\rm Im}\te$, and since $M$ is invertible, $M{\rm Im}\te={\rm Im}\te$ and $M^T\ker\te=\ker\te$. In the Euclidean case, $M^T=M^{-1}$, and we have $M\ker\te=\ker\te$. In the Lorentzian case, $M^T=\eta M^{-1}\eta$. But the metric $\eta$ commutes with $\te$, such that also in this case we arrive at $M\ker\te=\ker\te$.

	So for both signatures, $M$ decomposes as a direct sum $M=M_1\oplus M_2\in\B(\ker\te)\oplus\B(\rm Im\te)$, and we define $g_1:=(x,M_1\oplus1)$, $g_2:=(0,1\oplus M_2)$. Then  $g_1g_2=(x,M_1\oplus M_2)=g$, and as the slice $S$ contains the full image of $\te$, we have $g_2 S=S$ and consequently $g_1S=g_1g_2S=gS$. Furthermore, $g_1\in E_\te(d)$ (respectively $g_1\in\PG_\te(d)$) as $M_1\oplus1$ commutes with $\te=0\oplus\vartheta$.
\end{proof}

\begin{proposition}\label{proposition:E_te}
	\begin{enumerate}
		\item The net $(\E_\te,\Zyl_\te,\alpha^{\E,\te})$ satisfies the time zero condition with the Rieffel-deformed time zero $C^*$-algebras
		\begin{align}
			\E_{\te,0}(S)
			=
			\E_{0,\te}(S)
			\,,\qquad
			S\in\Slic_\te
			\,.
		\end{align}
		\item The restriction of $\sigma$ to $\E^\infty$ extends to a reflection positive functional $\sigma^\te$ on $\E_\te$.
	\end{enumerate}
\end{proposition}
\begin{proof}
	{\em a)} As vector spaces, the smooth time zero algebras of the net $\E_\te$ are, $S\in\Slic_\te$,
	\begin{align*}
		\E_{\te,0}(S)^\infty
		=
		\E^\infty\cap \bigcap_{Z\in\Zyl_\te\atop Z\supset S} \E_\te(Z)
		=
		\bigcap_{Z\in\Zyl_\te\atop Z\supset S} \E_\te(Z)^\infty
		=
		\bigcap_{Z\in\Zyl_\te\atop Z\supset S} \E(Z)^\infty
		=
		\E_0(S)^\infty
		\,,
	\end{align*}
	where we have used that $\E_\te(Z)^\infty=\E(Z)^\infty$ as vector spaces. Thus the closure in the norm $\|\cdot\|_\te$ gives $\E_{\te,0}(S)=\E_{0,\te}(S)$. By assumption, the undeformed net $\E$ satisfies the time zero condition, that is, the time zero algebras $\E_0(S)$ generate the cylinder algebras,
	\begin{align}\label{eq:EZgenerated}
		\E(Z)
		=
		\bigvee_{S\in\Slic_\te \atop g\in E_0(d), gS\subset Z} \alpha^\E_g(\E_0(S))
		\,,
	\end{align}
	where $\bigvee$ denotes the generated $C^*$-algebra. According to Lemma \ref{lemma:SlicesAndCylinders} in its Euclidean version, the transformations $g\in E_0(d)$, $gS\subset Z$, which appear here, split as $g=g_1g_2$ with $g_1\in E_\te(d)$ and $g_2S=S$. In view of the covariance of the undeformed net, this implies that $\alpha^\E_{g_2}$ restricts to an automorphism of $\E_0(S)$, {\em i.e.}, $\alpha_{g_2}^\E(\E_0(S))=\E_0(S)$. Thus \eqref{eq:EZgenerated} also holds if we restrict to $g\in E_\te(d)\subset E_0(d)$.

	To make the transition to the deformed $C^*$-algebras, we first restrict to the smooth time zero algebras $\E_0(S)^\infty\subset\E_0(S)$, and consider the ${}^*$-algebra $\hat{\E}(Z)$ generated by all $\alpha^\E_g(\E_0(S)^\infty)$, where $S$ runs over $\Slic_\te$ and $g$ over $E_\te(d)$ such that $gS\subset Z$. As vector spaces, $\E_{0,\te}(S)^\infty=\E_0(S)^\infty$, and also the automorphisms $\alpha_g^{\E,\te}$ and $\alpha_g^\E$, $g\in E_\te(d)$, coincide on $\E^\infty$. Hence $\alpha^{\E,\te}_g(\E_{0,\te}(S)^\infty)=\alpha^{\E}_g(\E_{0}(S)^\infty)$, and as this algebra is $\|\cdot\|_\te$-dense in $\alpha^{\E,\te}(\E_{0,\te}(S))$, it follows that the $\|\cdot\|_\te$-closure of $\hat{\E}(Z)$ coincides with $\E_\te(Z)$. In particular, we have the claimed time zero property
	\begin{align}
		\E_\te(Z)
		=
		\bigvee_{S\in\Slic_\te \atop g\in E_\te(d), gS\subset Z} \alpha^{\E,\te}_g(\E_{0,\te}(S))
		\,.
	\end{align}

	{\em b)} According to Lemma \ref{lemma:RieffelProduct2} {\em b)}, the restriction of the continuous linear translationally invariant functional $\sigma$ to $\E^\infty$ extends to a $\|\cdot\|_\te$-continuous functional $\sigma^\te$ on $\E_\te$. Since $\sigma$ is $E_0(d)$-invariant, it follows that this extension is invariant under the extension $\alpha^{\E,\te}$ of $\alpha^\E|_{E_\te(d)}$ from $\E^\infty$ to $\E_\te$. The continuity $E_\te(d)\ni g\mapsto \sigma(A\alpha^\E_g(B))$, $A,B\in\E_\te$, is then clear.

	It remains to check reflection positivity. By the translational invariance of $\sigma$, we have for smooth $A,B\in\E^\infty$ by Lemma \ref{lemma:RieffelProduct2} {\em b)}
	\begin{align*}
		\sigma(\iota_e(A^*)\times_\te B)
		=
		\sigma(\iota_e(A^*)B)
		\,,
	\end{align*}
	and hence in particular $\sigma(\iota_e(A^*)\times_\te A)\geq0$ for $A\in\E_{>}^\infty$. In view of the $\|\cdot\|_\te$-continuity of $\sigma^\te$ and $\iota^\te_e$, this positivity extends to $\E_{>,\te}$.
\end{proof}

We are now in the position to apply the Wick rotation from Section \ref{section:symmetries} to the net $(\E_\te,\Zyl_\te,\alpha^{\E,\te})$ and the reflection positive functional $\sigma^\te$. That is, we can construct a Minkowski Hilbert space $\Hil_\te^\M$ and a unitary representation $U_\te^\M$ of $\PG_\te(d)$ by analytic continuation of the virtual representation $V_\te$ of $E_\te(d)$, and a representation $\pi_{\sigma^\te,\te}$ of the time zero algebras on $\Hil_\te^\M$. We use an extra subscript $\te$ here to distinguish these data from similar objects introduced below. The generated $\PG_\te(d)$-covariant net of von Neumann algebras will be denoted $(\M_\te,\Zyl_\te,{\rm Ad}\,U_\te^\M)$, corresponding to the lower left corner in our commutative diagram. Recall that by construction, the time zero von Neumann algebras of this net are \eqref{def:net-M}
\begin{align}
	\M_{\te,0}(S)
	=
	\pi_{\sigma^\te,\te}(\E_{\te,0}(S))''
	\,,\qquad
	S\in\Slic_\te\,.
\end{align}
Note that the smooth time zero subalgebras are
\begin{align}\label{eq:MTimeZeroSmooth}
		\M_{\te,0}(S)^\infty=\pi_{\sigma^\te,\te}(\E_{\te,0}(S))^\infty=\pi_{\sigma^\te,\te}(\E_0(S)^\infty)\,,
\end{align}
where the last equality follows from the strong continuity of $\alpha^{\E,\te}|_{\Rl^d}$ on $\E$ and the continuity of $\pi_{\sigma^\te,\te}$. Since $\C^\infty\subset\B(\Hil)$ is weakly dense, we have $\M_{\te,0}(S)=(\M_{\te,0}(S)^\infty)''$.
Also note that in view of the deformation of the product, the time zero representation acts according to \eqref{def:PiSigma},
\begin{align}
	\pi_{\sigma^\te,\te}(A)[B]_{\sigma^\te,\te}
	=
	[A\times_\te B]_{\sigma^\te,\te}
	\,,\qquad
	B\in\E_{\te,>},\; A\in\E_{\te,0}(S),\; S\in\Slic_\te\,.
\end{align}
\\
\\
Passing to the other side of the diagram, we now start again with  the original data $(\E,\OOO,\alpha^\E)$, $\sigma$, but first apply the Hilbert space construction and continuation of virtual representation as in Section \ref{section:symmetries}. As we are using the undeformed data corresponding to $\te=0$ here, we will denote all quantities derived here with a subscript $0$, {\em i.e.}, write $\Hil^\M_0,V_0,U^\M_0, \pi_{\sigma,0}, \Om_0$.

According to Section \ref{section:symmetries}, the Wick rotated von Neumann algebras \eqref{def:net-M}
\begin{align}
	\M(\OO)
	=
	\bigvee_{K\perp e\atop g\in\PG_0(d),\,gK\subset \OO}
		\alpha^\M_g(\E_0(K))''
\end{align}
form a $\PG_0(d)$-covariant net $(\M,\OOO,\alpha^\M)$ with $\alpha^\M:={\rm Ad}U^\M_0$. This net is now deformed by warped convolution, where in \eqref{eq:Warping}, we take the Minkowski inner product $(\,\cdot\,,\,\cdot\,)^\M$ in the oscillatory integrals and the translations $x\mapsto U^\M_0(x,1)$. As deformation parameter, we can use the same matrix $\te$ as before, since this is antisymmetric w.r.t. the Minkowski inner product as well because the time direction $e$ lies in its kernel. Passing again to a subnet over cylinder regions, we define the von Neumann algebras
\begin{align}\label{eq:MteZ}
	\widetilde{\M}_\te(Z)
	:=
	\{A_\te\,:\,A\in\M(Z)^\infty\}''
	\,,\qquad
	Z\in\Zyl_\te\,,
\end{align}
and $\PG_\te(d)$-action
\begin{align}
	\tilde{\alpha}^\M
	:=
	{\rm Ad}\,U^\M_0|_{\PG_\te(d)}
	\,.
\end{align}

\begin{proposition}
	$(\widetilde{\M}_\te,\Zyl_\te,\tilde{\alpha}^\M)$ is a $\PG_\te(d)$-covariant net of von Neumann algebras which satisfies the time zero condition with the time zero von Neumann algebras
	\begin{align}\label{eq:MtildeTimeZeroAlgebras}
		\widetilde{\M}_{\te,0}(S)
		=
		{\pi_\sigma(\E_0(S)^\infty)_\te}''
		=
		\{\pi_\sigma(A)_\te\,:\,A\in\E_0(S)^\infty\}''
		\,,\qquad
		S\in\Slic_\te
		\,.
	\end{align}
\end{proposition}
\begin{proof}
	It is clear from the definition \eqref{eq:MteZ} that $\Zyl_\te\ni Z\mapsto\widetilde{\M}_\te(Z)$ is a net of von Neumann algebras on $\Hil^\M_0$. According to Lemma \ref{lemma:Warping} {\em b)}, the transformations $g$ in the reduced Poincar\'e group $\PG_\te(d)$ satisfy $U^\M_0(g)A_\te U^\M_0(g)^{-1}=(U^\M_0(g)A U^\M_0(g)^{-1})_\te$ for smooth $A$. Taking into account the covariance of the undeformed net $\M$, this implies
	\begin{align*}
		U^\M_0(g)\widetilde{\M}_\te(Z)^\infty U^\M_0(g)^{-1}
		=
		\widetilde{\M}_\te(gZ)^\infty
		\,,\qquad
		Z\in\Zyl_\te
		\,,\;
		g\in\PG_\te(d)
		\,.
	\end{align*}
	Because of the shape of the cylinder regions, we have $A\times_\te B\in\M(Z)^\infty$ for $A,B\in\M(Z)^\infty$, {\em cf.} \eqref{eq:RieffelProductReduced}. In view of Lemma \ref{lemma:Warping} {\em a)} and Lemma \ref{lemma:WarpingAndRieffel} {\em a)}, this implies that the set $\{A_\te\,:\,A\in\M(Z)^\infty\}$ is a ${}^*$-algebra. Hence the double commutant in \eqref{eq:MteZ} amounts to just taking the weak closure, and covariance of $\widetilde{\M}_\te$ follows.

	The smooth time zero algebras of the net $\widetilde{\M}_\te$ are, $S\in\Slic_\te$,
	\begin{align*}
		\widetilde{\M}_{\te,0}(S)^\infty
		=
		\bigcap_{Z\in\Zyl_\te\atop Z\supset S} \widetilde{\M}_\te(Z)^\infty
		=
		\M_{0,\te}(S)^\infty
		\,.
	\end{align*}
	Now $\M_0(S)=\pi_\sigma(\E_0(S))''=(\pi_\sigma(\E_0(S))^\infty)''$ because $\C^\infty\subset\B(\Hil)$ is weakly dense, and using $\pi_\sigma(\E_0(S))^\infty=\pi_\sigma(\E_0(S)^\infty)$, we arrive at $\widetilde{\M}_{\te,0}(S)={\pi_\sigma(\E_0(S)^\infty)_\te}''$.

	In the undeformed situation, the time zero condition holds, and making use of $(\M_0(S)^\infty)''=\M_0(S)$, it is not difficult to see that $\M(Z)$ is the smallest von Neumann algebra containing all $U_0^\M(g)\M_0(S)^\infty U_0^\M(g)^{-1}$, where $g\in \PG_0(d)$, $S\in\Slic_\te$ such that $gS\subset Z$. As in the proof of Proposition \ref{proposition:E_te} {\em a)}, we can apply (the Lorentzian version of) Lemma \ref{lemma:SlicesAndCylinders} to conclude that restriction to $g\in\PG_\te(d)\subset\PG_0(d)$ does not change the generated von Neumann algebra. After passing to warped convolution time zero algebras $\widetilde{\M}_{\te,0}(S)^\infty=\pi_\sigma(\E_0(S)^\infty)_\te$, we obtain
	\begin{align*}
		\widetilde{\M}_\te(Z)
		&=
		\bigvee_{S\in\Slic_\te,\,g\in E_\te(d)\atop{gS\subset Z}}
			\left(
				U_0^\M(g) \pi_\sigma(\E_0(S)^\infty)U_0^\M(g)^{-1}
			\right)_\te
		\\
		&=
		\bigvee_{S\in\Slic_\te,\,g\in E_\te(d)\atop{gS\subset Z}}
				U_0^\M(g) \widetilde{\M}_{\te,0}(S)U_0^\M(g)^{-1}
		\,,
	\end{align*}
	which is the claimed time zero condition.
\end{proof}

The main result of this section is the following theorem.
\begin{theorem}
	The two nets $(\widetilde{\M}_\te,\Zyl_\te,{\rm Ad}U^\M_0|_{\PG_\te(d)})$ and $(\M_\te,\Zyl_\te,{\rm Ad}U^\M_\te)$ are isomorphic, i.e., there exists a unitary $W:\Hil^\M_0\to\Hil^\M_\te$ such that
	\begin{align}
		W\Om_0
		&=
		\Om_\te
		\label{eq:WOmega}
		\,,\\
		WU^\M_0(g)W^*
		&=
		U^\M_\te(g)
		\,,\qquad
		g\in \PG_\te(d)
		\label{eq:U0Utheta}
		\,,\\
		W\widetilde{\M}_\te(Z)W^*
		&=
		\M_\te(Z)
		\,,\qquad
		Z\in\Zyl_\te\,.
		\label{eq:IntertwineNets}
	\end{align}
\end{theorem}
\begin{proof}
	We first relate the GNS-type Hilbert spaces $\Hil^\M_\te$, $\Hil^\M_0$ in close analogy to \cite[Prop.~2.3]{Lechner:2011}: For $A,B\in\E^\infty_{>}$, we have by the translational invariance of $\sigma$ and Lemma~\ref{lemma:RieffelProduct2}~{\em b)}
	\begin{align*}
		\langle[A]_\sigma,[B]_\sigma\rangle_{\Hil^\M_0}
		=
		\sigma(\iota_e(A^*)B)
		=
		\sigma^\te(\iota_e(A^*)\times_\te B)
		=
		\langle[A]_{\sigma^\te,\te},[B]_{\sigma^\te,\te}\rangle_{\Hil^\M_\te}
		\,.
	\end{align*}
	This shows that the map $W_0:\E^\infty_{>}/(\NN_0\cap\E_{>}^\infty)\to \E^\infty_{>}/(\NN_\te\cap\E_{>}^\infty)$,
	\begin{align*}
		W_0[A]_\sigma
		:=
		[A]_{\sigma^\te,\te}
		\,,\qquad
		A\in\E^\infty_{>}\,,
	\end{align*}
	is well-defined and isometric. Since its domain and range are dense, we can extend it to a unitary $W:\Hil^\M_0\to\Hil^\M_\te$. Clearly $W$ satisfies $W\Om_0=W[1]_\sigma=[1]_{\sigma^\te,\te}=\Om_\te$ \eqref{eq:WOmega}.

	On $\Hil_\te$, we thus have two virtual representations $g\mapsto V_\te(g)$ and $g\mapsto WV_0(g)W^*$ of the reduced Euclidean group $E_\te(d)$. For $g$ in a sufficiently small neighborhood of the identity, these representations act according to \eqref{def:V}, $A\in\E^\infty_>$,
	\begin{align*}
		V_\te(g)[A]_{\sigma^\te,\te}
		&=
		[\alpha^\E_g(A)]_{\sigma^\te,\te}
		\,,\\
		WV_0(g)W^*[A]_{\sigma^\te,\te}
		&=
		WV_0(g)[A]_{\sigma}
		=
		W[\alpha^\E_g(A)]_\sigma
		=
		[\alpha^\E_g(A)]_{\sigma^\te,\te}
		\,,
	\end{align*}
	that is, they coincide. After analytic continuation to unitary representations of $\PG_\te(d)$, this implies \eqref{eq:U0Utheta}.

	To show that $W$ also intertwines the nets, it is sufficient to consider the time zero algebras since both $\M_\te$ and $\widetilde{\M}_\te$ satisfy the time zero condition and are generated from their time zero data by $W$-equivalent representations of $\PG_\te(d)$.

	Both $\M_\te(Z)$ and $\widetilde{\M}_\te(Z)$ are generated as von Neumann algebras from their respective smooth time zero algebras ({\em cf.} \eqref{eq:MTimeZeroSmooth}, \eqref{eq:MtildeTimeZeroAlgebras}),
	\begin{align*}
		\M_{\te,0}(S)^\infty
		&=
		\pi_{\sigma^\te,\te}(\E_0(S)^\infty)
		=
		\{\pi_{\sigma^\te,\te}(A)\,:\,A\in\E_0(S)^\infty\}
		\,,
		\\
		\widetilde{\M}_{\te,0}(S)^\infty
		&=
		\pi_\sigma(\E_0(S)^\infty)_\te
		=
		\{\pi_\sigma(A)_\te\,:\,A\in\E_0(S)^\infty\}
		\,.
	\end{align*}
	Comparing these two time zero algebras, it becomes apparent that it is sufficient to show that $W$ intertwines $\pi_{\sigma}(A)_\te$ and $\pi_{\sigma^\te,\te}(A)$ for $A\in\E_0(S)^\infty$, as then \eqref{eq:IntertwineNets} follows by continuity.

	With $A\in\E_0(S)^\infty$ and $B\in\E^\infty_>$, we compute using Lemma \ref{lemma:Warping} {\em c)}
	\begin{align*}
		W\pi_{\sigma}(A)_\te W^*[B]_{\sigma^\te,\te}
		&=
		W\pi_{\sigma}(A)_\te [B]_{\sigma}
		\\
		&=
		W\pi_{\sigma}(A)_\te \pi_\sigma(B)\Om_0
		\\
		&=
		W\pi_{\sigma}(A)_\te \pi_\sigma(B)_\te\Om_0
		\,.
	\end{align*}
	Note that the warped convolutions $\pi_\sigma(A)_\te$, $\pi_\sigma(B)_\te$ are build here with the representation $U^\M_0$ and the Minkowski inner product in the oscillatory integral \eqref{eq:Warping}. But as $\te e=0$, only spatial translations along $x\perp e$ enter \eqref{eq:WarpingReduced}. For $p,x\perp e$, the Euclidean and Minkowski inner products coincide\footnote{This is due to our choice of signature  $(-1,+1,...,+1)$. For the choice with the opposite signs one has to use the negative noncommutativity $-\te$ on the Minkowski side.}, $(p,x)^\E=(p,x)^\M$. Furthermore, for $x\perp e$, the unitaries $U^\M_0(x,1)$ implement $\alpha^\E_x$. So we can use Lemma \ref{lemma:WarpingAndRieffel} {\em b)} stating that $A\mapsto \pi_\sigma(A)_\te$ is an $\alpha$-covariant representation of the Rieffel-deformed $C^*$-algebra $\E_{0,\te}(S)$, and again Lemma \ref{lemma:Warping} {\em c)} to compute further
	\begin{align*}
		W\pi_{\sigma}(A)_\te W^*[B]_{\sigma^\te,\te}
	       &=
		W\pi_{\sigma}(A \times_\te B)_\te\Om_0
		\\
		&=
		W\pi_{\sigma}(A \times_\te B)\Om_0
		\\
		&=
		W[A \times_\te B]_{\sigma}
		\\
		&=
		[A \times_\te B]_{\sigma^\te,\te}
		\\
		&=
		\pi_{\sigma^\te,\te}(A)[B]_{\sigma^\te,\te}
		\,.
	\end{align*}
	As all operators appearing here are bounded and $\{[B]_{\sigma^\te,\te}\,:\,B\in\E^\infty_>\}\subset\Hil^\M_\te$ is dense, we obtain $W\pi_{\sigma}(A)_\te W^*=\pi_{\sigma^\te,\te}(A)$ by continuity.
\end{proof}

The unitary $W$ implements the net isomorphism in the diagram and completes the proof that the operations of Wick rotating and deforming commute.

It should be observed that the net $\M_\te\cong\widetilde{\M}_\te$ is non-local, {\em i.e.}, does {\em not} satisfy $[\M_\te(Z_1),\M_\te(Z_2)]=\{0\}$ for spacelike separated regions $Z_1,Z_2$, even if $\E$ is abelian. For deformed theories satisfying certain remnants of locality/causality, one needs more than one deformation parameter, {\em cf.} \cite{BuchholzLechnerSummers:2010}.

%% file: conclusion.tex
\section{Conclusions and Outlook}\label{section:conclusions}

Quantum field theories on noncommutative spaces differ significantly in many aspects from usual quantum field theories on Euclidean or Minkowski space. This is true in particular for their localization properties, which are expected to be weaker than in the commutative case because of space--time uncertainty relations. But depending on the model used, also the symmetry and covariance properties of field theories on noncommutative spaces are often much weaker than the full Euclidean/Poincar\'e covariance familiar from field theories on commutative space. Since both these features, covariance and locality, are important for establishing the classical ``Wick rotation'' relating fields in Lorentzian and Euclidean signature \cite{StreaterWightman:1964}, one might doubt such a connection exists in the noncommutative case. This is also the impression one gets when working on the level of perturbative renormalization, where field theories on Euclidean and Minkowski space behave quite differently \cite{Bahns:2009}.

These differences are less pronounced in the special case that the time variable still commutes with the spatial coordinates, although also in that case, the usual analytic continuation of $n$-point Wightman functions to imaginary time does not reach all the Euclidean points where the Schwinger functions are defined, and the locality and covariance properties are not better than in the case of noncommutative time. In the present work, we have shown how one can relate field theories on noncommutative spaces with commuting time despite these problems. By concentrating on analytic continuations of symmetry group representations instead of $n$-point functions, and making use of the algebraic setting of quantum field theory, we were able to establish a tight and natural connection between Euclidean and Minkowski space Moyal-deformed field theories. Such a connection might be expected by considering the analytic continuation of Wightman $n$-point functions of deformed quantum field theories \cite{Soloviev:2006,GrosseLechner:2008} to the Euclidean points in the forward tube, where they match with their deformed Euclidean counterparts \cite{Bahns:2009}. However, most Euclidean points do not lie in this domain, but only in an enlarged tube which in the commutative case can be constructed with the help of Lorentz symmetry and locality \cite{StreaterWightman:1964}. The precise relation of our present analysis of nets over cylinder regions to analyticity properties of $n$-point functions and the discussion of explicit model theories will be presented elsewhere. In case the time coordinate does not commute with the spatial coordinates, however, the approach taken here can not be used without major modifications.

In the present work we concentrated on models of quantum fields with restricted (Euclidean or Poincar\'e) symmetry. As is well known, one can also formulate fully covariant models on Moyal space (see, for example, \cite{DoplicherFredenhagenRoberts:1995,GrosseLechner:2007,Piacitelli:2010}) by including not a single noncommutativity $\te$, but rather a full orbit of these matrices in the model. In the Minkowski case, models formulated in this extended setting are not only fully Poincar\'e covariant, but also show interesting remnants of locality (``wedge locality'') \cite{GrosseLechner:2007,BuchholzLechnerSummers:2010}. We postpone a detailed discussion of the Wick rotation of fully covariant models and their residual locality properties to a future investigation.